\documentclass[aps,prx,twocolumn,footinbib,
showpacs,showkeys,longbilbliography]{revtex4-2}

\usepackage{graphicx}
\usepackage{indentfirst}
\usepackage{physics}
\usepackage{braket}
\usepackage{float}
\usepackage{amsmath}
\usepackage{amssymb}
\usepackage{verbatim}
\usepackage{wasysym}
\usepackage{epstopdf}
\usepackage{CJK}
\usepackage{esint}
\usepackage{color}
\usepackage{xcolor}
\usepackage{subfigure}
\usepackage{amsfonts}
\usepackage{footmisc}
\usepackage{scrextend}
\usepackage{multirow}
\usepackage{mathtools}
\usepackage{xr-hyper}
\usepackage[hyperfootnotes=false]{hyperref}

\usepackage[english]{babel}
\usepackage{url}
\usepackage{bm}
\definecolor{darkblue}{rgb}{0,0,0.5}
\hypersetup{
colorlinks=true,
linkcolor=black,
filecolor=black,
citecolor=darkblue,
urlcolor=darkblue,
}

\newcommand{\defeq}{\vcentcolon=}

\DeclareMathOperator*{\argmin}{arg\,min}

\newcommand\mc[1]{\mathcal{#1}}

\newcommand\bs[1]{\boldsymbol{#1}}

\newtheorem{theorem}{Theorem}
\newtheorem{lemma}{Lemma}
\newtheorem{corollary}{Corollary}
\newtheorem{proposition}{Proposition}
\newtheorem{defin}{Definition}
\newenvironment{proof}[1][Proof]{\noindent\textbf{#1.} }{\ \rule{0.5em}{0.5em}}

\begin{document}

\title{Analytical Methods for High-Rate Global Quantum Networks}

\author{Cillian Harney}
\email{cth528@york.ac.uk}
\author{Stefano Pirandola}
\email{stefano.pirandola@york.ac.uk}
\affiliation{Department of Computer Science, University of York, York YO10 5GH, United Kingdom}
\begin{abstract}
The development of a future, global quantum communication network (or quantum internet) will enable high rate private communication and entanglement distribution over very long distances. 
However, the large-scale performance of ground-based quantum networks (which employ photons as information carriers through optical-fibres) is fundamentally limited by fibre quality and link length, with the latter being a primary design factor for practical network architectures.
While these fundamental limits are well established for arbitrary network topologies, the question of how to best design global architectures remains open. 
In this work, we introduce a large-scale quantum network model called weakly-regular architectures. Such networks are capable of idealising network connectivity, provide freedom to capture a broad class of spatial topologies and remain analytically treatable. This allows us to investigate the effectiveness of large-scale networks with consistent connective properties, and unveil critical conditions under which end-to-end rates remain optimal. Furthermore, through a strict performance comparison of ideal, ground-based quantum networks with that of realistic satellite quantum communication protocols, we establish conditions for which satellites can be used to outperform fibre-based quantum infrastructure; {rigorously proving the efficacy of satellite-based technologies for global quantum communications.}
\end{abstract}
\maketitle

\section{Introduction}

Advancements in quantum information science will have a profound impact on society \cite{Mike_Ike, WatrousTxt, Holevo19, GoogleSyc}. In particular, the overarching trajectory of quantum communication technologies is towards a global quantum communication network: A quantum internet \cite{KimbleQI,UniteQInt,RazaviQNet,AdvCrypt}. This will facilitate high rate, provably secure communication and globally distributed quantum information processing with radical implications for science, technology and beyond.\par

The current, most promising point-to-point quantum communication protocols (where two parties are connected directly via a quantum channel) are based on continuous variable (CV) quantum systems \cite{SerafiniCV,GaussRev,BraunsteinVL,RalphCV} (such as bosonic modes). CV protocols achieve high performance and are compatible with current telecommunication infrastructure based upon optical-fibre connections, emphasising their near-term feasibility.
However, the laws of quantum mechanics prohibit the ability to simultaneously achieve high rates and long distances, a fundamental law captured by the Pirandola-Laurenza-Ottaviani-Banchi (PLOB) bound \cite{PLOB}. This describes the absolute maximum rate that two parties may transfer quantum states, distribute entanglement, or establish secret-keys over a bosonic lossy channel (optical-fibres) equal to $-\log_{2}(1-\eta)$ bits per channel use, where $\eta$ is the channel transmissivity \cite{PLOB, PirPatron09}. \par

Overcoming this point-to-point limitation requires the use of quantum repeaters or more generally the construction of quantum networks. Combining tools from classical network theory \cite{SlepianNets,CoverThomas, TanenbaumNets, GamalNets} with the PLOB bound, ultimate limits have also been established for the end-to-end capacities of quantum networks \cite{End2End}. These results confirm that the PLOB bound can be beaten via quantum networking, facilitating high rate communication at longer ranges. While such bounds are easily expressed in full generality for arbitrary network topologies, their practical assessment requires the specification of an architecture. Questions of network topology have been recently considered via the statistical study of complex, random quantum networks \cite{BiamonteNets,BritoRandQNets, QuntaoRandQNets, ZhangQInt}, which reveal insightful phenomena associated with large-scale network properties. Studies of this kind are extremely valuable and help to unveil important guidelines for the development of future quantum networks.

Nonetheless, such analyses are not easy and require significant numerical effort in order to reveal key network properties, e.g.~critical network densities or maximum fibre-lengths. There is a demand for versatile, \textit{analytical} tools which allow for the efficient benchmarking of quantum networks and can motivate the construction of large-scale topologies.

Meanwhile, ground-based fibre channels are not the only conduits available for global quantum communications. A rival infrastructure that may prove superior to fibre-based networks at global distances is Satellite Quantum Communication (SQC) \cite{FS,AdvSQC,SQC,Micius2017,RenSQC,LiaoSQC,VillarSQC,Micius2020}. SQC exploits ground-to-satellite communication channels to overcome the fundamental distance limitations offered by fibre/ground-based mechanisms. A satellite in orbit around the Earth may act as a \textit{dynamic} repeater that physically passes over ground-based users and distributes very long-range entanglement/secret-keys. The ability to exploit a free-space connection with a satellite also carries the possibility of substantially more transmissive channels than optical-fibre, making it ideal for global communication protocols.\par

The following critical questions emerge: Can we develop analytical tools which allows us to place limits on the end-to-end performance of large-scale quantum networks?
And are fibre-based networks truly the best way to achieve long-distance quantum communication? The goal of this work is to make progress with these challenges. 

Utilising ideas from quantum information theory, classical networks and graph theory \cite{WilsonGraph}, we investigate ideal architectures based on the property of weak-regularity. {Weakly-Regular Networks (WRNs) simultaneously (\textit{i}) idealise network connectivity, (\textit{ii}) provide sufficient freedom to capture a broad class of spatial topologies and (\textit{iii}) remain analytically treatable so that critical network properties can be rigorously studied. This results in a design with desirable qualities which can efficiently and effectively provide insight for realistic structures. We show that quantum WRNs employing multi-path routing admit remarkably accessible and achievable upper-bounds on the end-to-end network capacity. This allows for a characterisation of the ideal performance of a fibre-based quantum internet with respect to essential properties such as maximum channel length and nodal density. 

Our exact, analytical results provide an immediate pathway to perform comparisons of SQC with global ground-based quantum communications. We study the average number of secret bits per day that can be distributed between two remote stations, using large-scale quantum fibre-networks or a single satellite repeater station in orbit. Our findings rigorously prove the superiority of satellite-based quantum repeaters for global quantum communications, reveal the constraints associated with fibre-based networks and the enormous resource demands required to overcome achievable rates offered by a single satellite. These results further motivate the study of SQC and its key role within a future quantum internet.

The remainder of this paper is structured as follows: In Section~\ref{sec:NetDesign} we introduce preliminary notions of quantum networks, optimal end-to-end performance and ideal properties of large-scale network designs. In Section~\ref{sec:WRNs} we discuss how the optimal performance of quantum WRNs can be analytically bounded with respect to network properties and apply these methods to bosonic lossy-networks. Section~\ref{sec:Comparison} then compares the performance of global fibre-networks with rates that are achievable by SQC, assessing the advantages associated with each infrastructure, followed by concluding remarks.

\section{Quantum Network Design \label{sec:NetDesign}}

\subsection{Basics of Quantum Networks}

A quantum network can be described as a finite, undirected graph $\mathcal{N} = (P,E)$ where $P = \{\bs{x}_i\}_{i}$ is a set of all nodes (points/vertices) on the graph and $E = \{e_i\}_{i}$ collects valid connections between pairs of nodes (edges). A network node refers to either a user-node, such as a potential end-user pair Alice $\bs{a}$ and Bob $\bs{b}$, or a repeater/relay-node. Each node $\bs{x}_i$ possesses a local register of quantum systems which can be altered and exchanged with connected neighbours. Any two nodes $\bs{x},\bs{y} \in P$ are connected via an undirected edge $e \defeq (\bs{x},\bs{y})\in E$ if there exists a quantum channel $\mathcal{E}_{\bs{x}\bs{y}}$ through which they may directly communicate. Since each edge is undirected, this may be a forward or backward channel. 

In the context of quantum networks, it is important to make a distinction between \textit{physical flow} and \textit{logical flow}. The logical flow of a quantum communication channel describes the direction in which entanglement, secret-keys, or quantum states are distributed from a node $\bs{x}$ to node $\bs{y}$ (or vice versa). The physical flow of quantum communication refers to the actual direction of quantum system exchange, i.e.~if quantum systems are physically sent in the direction $\bs{x}\rightarrow \bs{y}$ or $\bs{y}\rightarrow \bs{x}$. In a quantum network, these concepts can be completely decoupled. This may be due to the fact that the communication task has a symmetric objective i.e.~if Alice and Bob with to share a secret-key, they do not care \textit{who} initiates the exchange of quantum systems. However, it may also be thanks to quantum teleportation; it is always possible to ``reverse" the logical direction of communication by means of a teleportation protocol between Alice and Bob. 

The independence of physical and logical flow helps us to reliably describe a quantum network as an undirected graph. Any pair of connected network nodes can choose the physical direction in which they wish to exchange quantum systems and may always choose that which has the largest capacity. As a result, we never need to distinguish between forward or backward channels and represent each edge $(\bs{x},\bs{y})\in E$ by the best choice of quantum channel \cite{End2End}.

In a point-to-point communication setting, the logical flow of quantum information has a clear and obvious set of choices; Alice to Bob $\bs{a} \rightarrow \bs{b}$ or Bob to Alice $\bs{a} \leftarrow \bs{b}$. However, within a quantum network, a vast array of options emerge due to the various interconnections and possible paths that information may follow. To address this, users can devise an end-to-end routing strategy that facilitates communication between end-users. The two key classes of strategy are \textit{single-path} and \textit{multi-path} routing.\par

Single-path routing is the simplest network communication method, which utilises point-to-point communications in a sequential manner. An end-to-end route, $\omega$, is defined as a sequence of network edges which are able to connect a pair of end-users $\bs{a},\bs{b}\in P$, that is $\omega \defeq \{ (\bs{a},\bs{x}_1), (\bs{x}_1, \bs{x}_2), \ldots, (\bs{x}_N, \bs{b})\}.$
Quantum systems can be exchanged from node-to-node along this route, followed by LOCC operations after each transmission until eventually communication is established between the end-users. This kind of strategy is analogous to the use of a repeater-chain and network performance is determined by the strength of each link along an optimal end-to-end route. Yet, quantum information is significantly less robust than classical information and single-path routing may not be sufficiently resilient to network errors, or provide high enough rates.

A more powerful strategy is multi-path routing, which properly exploits the multitude of possible end-to-end routes available in a quantum network. In multi-path protocols, users may simultaneously utilise a number of unique routes $\{\omega_1,\omega_2,\ldots,\omega_M\}$ in an effort to enhance their end-to-end rate. A user may exchange an initially multi-partite quantum state with a number of neighbouring receiver nodes, who may each then perform their own point-to-multi-point exchanges along its unused edges. The exchange of quantum systems can be interleaved with adaptive network LOCCs in order to distribute secret correlations and this process continues until a multi-point interaction is carried out with the end-user. 

The optimal multi-path routing strategy operates in such a way that all channels in the network are used once per end-to-end transmission. This is known as a \textit{flooding protocol} \cite{TanenbaumNets, GamalNets,End2End}; each node in the network performs quantum systems exchanges along all its available edges, resulting in non-overlapping point-to-multipoint exchanges between all network nodes. The ability to flood an entire network means that every possible end-to-end route between the end-users are fully explored, allowing them to achieve the optimal end-to-end rate. This greatly enhances the end-to-end performance of quantum networks.

\subsection{Optimal Performance and Flooding Capacities\label{sec:Performance}}

As discussed, the optimal end-to-end performance of a network is defined by its ability to perform flooding by using every edge in the network to achieve communication between a pair of end-users \cite{End2End}. Any communication protocol which does not flood the network utilises less resources and thus fewer end-to-end paths; hence no protocol can achieve a better end-to-end rate than flooding. This optimal performance is quantified by a \textit{flooding capacity} $\mc{C}^m(\bs{i}, \mc{N})$, which describes the optimal number of target bits (such as secret-bits or entanglement-bits) that can be transmitted between end-users per use of a flooding protocol.

An important graph-theoretic concept for quantifying network performance is that of \textit{cuts} and \textit{cut-sets}. Consider a network $\mc{N}=(P,E)$ with two remote end-users $\bs{a}, \bs{b} \in P$. We may collect this end-user pair into its own, compact object $\bs{i} = \{ \bs{a},\bs{b}\}$, which is simply a two-element subset of the collection of all network nodes. We define a cut $C$ as a bipartition of all network nodes $P$ into two disjoint subsets of nodes $(P_{\bs{a}}, P_{\bs{b}})$ such that the end-users become completely disconnected, $\bs{a} \in P_{\bs{b}}$ and $\bs{b} \in P_{\bs{b}}$, where $P_{\bs{a}}\cap P_{\bs{b}} =\varnothing$. A cut $C$ generates an associated cut-set; a collection of network edges $\tilde{C}$ which enforce the partitioning when removed,
\begin{equation}
\tilde{C} \defeq \{ (\bs{x},\bs{y}) \in E~|~\bs{a} \in P_{\bs{a}}, \bs{b}\in P_{\bs{b}}\}.
\end{equation}
Under the action of a cut, a network is successfully partitioned
\begin{equation}
\mc{N} = (P,E) \xrightarrow{\text{Cut: $C$}} (P,E \setminus \tilde{C}) = ( P_{\bs{a}}\cup P_{\bs{b}}, E \setminus \tilde{C} ),
\end{equation}
so that there no longer exists a path between $\bs{a}$ and $\bs{b}$. Network cuts play a key role in the derivation of end-to-end network rates and many network optimisation tasks can be reduced to an optimisation over all cuts with respect to single-edge/multi-edge properties.

Each channel in a network is associated with a single-edge capacity $\mc{C}_{\bs{xy}} \defeq \mc{C}(\mc{E}_{\bs{xy}})$ which describes the point-to-point communication quality between nodes. Hence, all networks have a single-edge capacity distribution 
$
\{\mc{C}_{\bs{xy}}\}_{(\bs{x},\bs{y})\in E}$ which informs the weights of the network graph. Consequently, the flooding capacity can be derived by solving the classical maximum-flow minimum-cut problem according to a single-edge capacity distribution.
The flooding capacity is found by locating the minimum-cut $C_{\min}$, which minimises the multi-edge capacity over all cut-sets \cite{End2End},
\begin{equation}
\mc{C}^m(\bs{i}, \mc{N}) \defeq \min_{C}\hspace{-2mm}\sum_{(\bs{x},\bs{y})\in \tilde{C}} \mc{C}_{\bs{xy}}. \label{eq:Flood_Cap}
\end{equation}

For general quantum networks with arbitrary capacity distributions and network structures, this problem requires a numerical treatment by solving the well known max-flow min-cut problem \cite{FordFlow,KarpFlow,OrlinFlow} to find $C_{\min}$. However, it is possible to express an intuitive, simpler upper-bound. Let us define the nodal-neighbourhood of a node $\bs{x}$ as
\begin{equation}
N_{\bs{x}} \defeq \{ \bs{y} ~|~(\bs{x},\bs{y})\in E\}.
\end{equation}
Then $N_{\bs{x}}$ is the collection of nodes to which $\bs{x}$ is connected. We can also define an edge-neighbourhood of $\bs{x}$ as all the edges which connect $\bs{x}$ to its neighbours,
\begin{equation}
E_{\bs{x}} \defeq \{ (\bs{x},\bs{y}) ~|~\bs{y}\in N_{\bs{x}}\}.
\end{equation}
It follows that one can always successfully partition the users $\bs{a}$ and $\bs{b}$ by collecting all of the edges in $E_{\bs{a}}$ or $E_{\bs{b}}$ into a cut-set. This effectively disconnects either of the nodes from the rest of the network, resulting in a valid cut and is true regardless of the network architecture. We may call this form of network cut \textit{user-node isolation}, denoted $C_{\text{iso}}$. 

Since this form of cut always exists, the multi-edge capacity associated with $C_{\text{iso}}$ can be used to upper-bound Eq.~(\ref{eq:Flood_Cap}). By performing nodal isolation on the end-user in $\bs{i}=\{\bs{a},\bs{b}\}$ which minimises its multi-edge capacity, we can write the bound
\begin{equation}
\mc{C}^m(\bs{i}, \mc{N}) \leq \mc{C}_{\mc{N}_{\bs{i}}}^m \defeq \min_{\bs{j} \in \bs{i}} \sum_{(\bs{x},\bs{y})\in E_{\bs{j}}} \mc{C}_{\bs{xy}}. \label{eq:Nodal_Iso}
\end{equation}
Here we have defined $\mc{C}_{\mc{N}_{\bs{i}}}^m$ as the \textit{min-neighbourhood capacity} which is generated by $C_{\text{iso}}$. Since $C_{\text{iso}}$ is a valid network-cut, the min-neighbourhood capacity is achievable.

\subsection{Ideal Properties of Large-Scale Networks}

An overarching goal of quantum network design is to achieve \textit{end-to-end distance independence}. That is, given a pair of end-users, the rate achievable between them is independent from their physical end-to-end separation and instead encoded into some properties of network link-lengths or nodal density. In quantum networks, distance independence is especially important as it bypasses the fundamental rate limitations associated with point-to-point communications imposed by the PLOB bound. Recent studies have shown that random fibre-network architectures which are explicitly conscious of link-lengths are capable of obtaining distance independence, e.g.~Waxman networks which are sufficiently dense \cite{WaxmanNets, QuntaoRandQNets}. These studies simultaneously suggest the shortcomings of classically-inspired network architectures (such as scale-free structures) to achieve distance independence in a quantum setting, even with large resources. 

It is imperative that quantum networks are constructed using quantum-motivated design choices. To facilitate end-to-end distance independence, ground-based quantum network architectures will need to be especially conscious of two main features; maximum link-length (which is often encoded into nodal density) and network connectivity (how well connected each node is within the structure). When appropriate restrictions are placed on permitted link-lengths, this can help to ensure strong end-to-end rates. 
Until now, the behaviour of maximum link-lengths with respect to network architecture has been excluded to numerical studies. For example, in Refs.~\cite{QuntaoRandQNets, BritoRandQNets} the authors study the necessary {nodal densities} required to ensure effective end-to-end rates; the higher the network density, the more likely that end-to-end communication can be mediated by shorter links, resulting in reliable rates. In this work, we reveal novel analytical tools which help to uncover the necessary resources for high performance quantum fibre-networks.

\section{Weakly-Regular Quantum Networks\label{sec:WRNs}}

\subsection{Weakly-Regular Networks}

The generality of complex architectures such as Waxman, Erd\H{o}s-R\'{e}nyi and scale-free networks render analytical investigation very difficult. For this reason, the investigation of quantum repeater technologies/protocols often relies heavily on numerics in order to study complex network performance. Otherwise, one is limited to the simpler setting of linear networks which can be assessed analytically, i.e.~repeater chains. The development of a common ground between these scenarios, where large-scale, highly-connected networks can be studied analytically is thus highly desirable. 

Here, we propose the use of \textit{weakly-regular} (WR) network architectures. Weak-regularity is a graph-theoretic concept which infers particular connectivity properties onto undirected graphs, $\mc{N} = (P,E)$. Most prominently,
each node $\bs{x} \in P$ in a WRN has the same degree, i.e.~if $\bs{x} \in P$ is connected to $k$ other nodes, then every node in the network is connected to exactly $k$ nodes. This infers \textit{regularity} and provides the core simplification from complex graphical designs (where regularity is seldom held).

The \textit{weak} element of weak-regularity is less obvious, but is similarly integral to our analyses. By definition, a \textit{strongly-regular} graph is a structure which adheres to very strict rules; it consists of $n$ nodes which are all $k$ regular, any pair of adjacent nodes (nodes which share an edge) share exactly $\lambda$ common neighbours and any pair of non-adjacent nodes (don't share an edge) share exactly $\mu$ neighbours. We call these positive, integer parameters $\lambda$, $\mu \in \mathbb{Z}^+$ the adjacent and non-adjacent commonalities respectively of any two nodes on the graph and are used to characterise how a graph is connected. The notion of strength in strong-regularity resides in the consistency of the values $\lambda$ and $\mu$ for all nodes across the graph. 
When strong-regularity is upheld, all of these requests result in a relatively small graph with impractical properties for large-scale network design. Consequently, if $\lambda$ and $\mu$ are allowed to take on a broader range of values, then the graph is \textit{weakly}-regular. Hence, weakness infers a looser characterisation of neighbour sharing between nodes.

\begin{figure}
\includegraphics[width=0.9\linewidth]{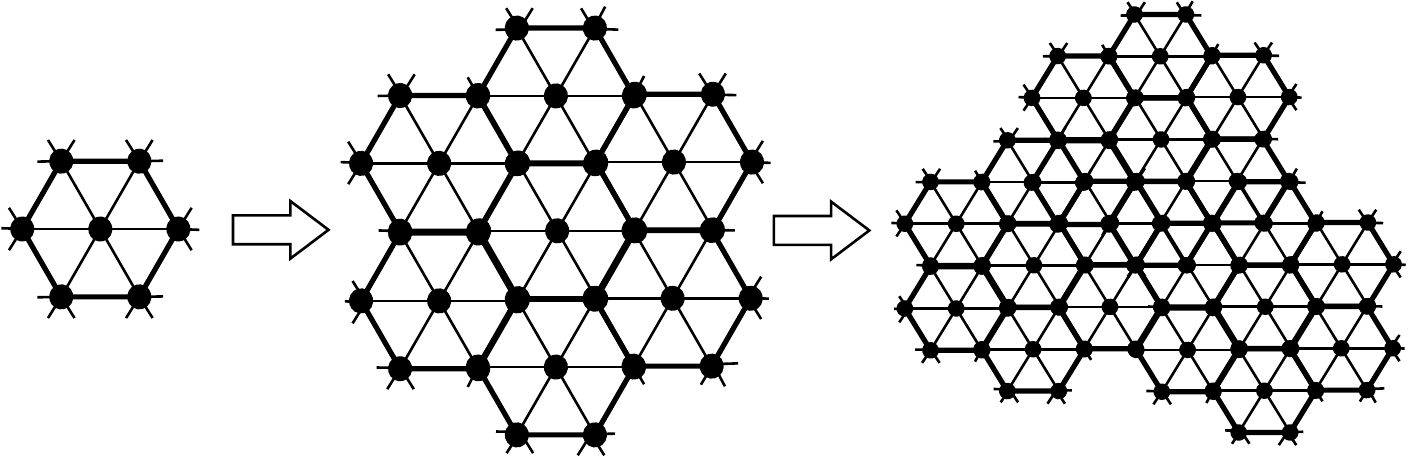}
\caption{A WRN cell can be concatenated many times to construct a large-scale network which is WR within some nodal boundary. Alternatively, outer edges can be looped in order to close the network so to satisfy weak-regularity everywhere (so there is no nodal boundary). In this example, $k = 6$ and there is only one unique adjacent commonality multi-set $\bs{\lambda}=\{2\}^{\cup 6}$ (we employ a superscript union notation to describe the repeated union of a single set, e.g.~$\{x\}^{\cup 3} = \{x\}\cup\{x\}\cup\{x\} = \{x,x,x\}$).}
\label{fig:WeakReg_Explain}
\end{figure}

In this work, we consider WR graphs for which the neighbour-sharing (commonality) properties of any network node can obey some \textit{spectrum} of values. In the context of studying end-to-end performance, it turns out that the most important commonality property is the adjacent commonality. Consider a node $\bs{x} \in P$ on a $k$-WR graph and its neighbourhood of adjacent nodes,
\begin{equation}
N_{\bs{x}} \defeq \{ \bs{y} ~|~(\bs{x},\bs{y})\in E\}.
\end{equation}
Then $N_{\bs{x}}$ is the collection of $k$ nodes to which $\bs{x}$ is connected. For any node on the network, we can define a bespoke \textit{adjacent commonality multi-set} (a modified set which can contain degenerate elements) which counts the number of neighbours shared between $\bs{x}$ and all its neighbours $\bs{y} \in N_{\bs{x}}$. More precisely, the adjacent commonality multi-set can be defined as
\begin{equation}
\bs{\lambda}_{\bs{x}} \defeq \big\{ |N_{\bs{x}} \cap N_{\bs{y}}| ~\big|~(\bs{x},\bs{y}) \in E\big\} = \{ \lambda_{\bs{x}}^{\bs{y}} \}_{\bs{y}\in N_{\bs{x}}},
\end{equation} 
where we denote $\lambda_{\bs{x}}^{\bs{y}} \defeq |N_{\bs{x}} \cap N_{\bs{y}}|$ as the number of common neighbours shared between the connected nodes $\bs{x}$ and $\bs{y}$. In our analyses, we consider graphs which are defined by a non-degenerate super-set of permitted adjacent commonality multi-sets,
\begin{equation}
\bs{\Lambda} \defeq \{ \bs{\lambda_x} ~|~\bs{x}\in {P}\},
\end{equation} 
so that the adjacent commonality multi-set of any node in the network $\bs{x}\in P$ belongs to the set $\bs{\lambda}_{\bs{x}} \in \bs{\Lambda}$. 

In summary, we are able to define a $(k,\bs{\Lambda})$-WRN as a class of network for which all nodes have a constant degree equal to $k$ and for which their neighbour-sharing properties satisfy $\bs{\lambda}_{\bs{x}} \in \bs{\Lambda}$. While $k$ and $\bs{\Lambda}$ impose connectivity constraints, WRNs that belong to this class are free to adopt a vast range of topological or spatial configurations. Furthermore, we avoid explicit references to the number of network nodes $n$. Instead, $n$ is encoded into properties of the network such as the nodal density $\rho_{\mc{N}}$ which defines the average number of network nodes per unit of area.

While these definitions may seem overly abstract, they introduce a remarkably versatile way to analytically describe interesting and useful network structures. For example, it is easy to construct a WR \textit{network cell}; a collection of nodes connected a particular graphical structure, which when concatenated (or ``stitched") together will result in a large-scale network which obeys weak-regularity within some nodal boundary. This permits the analytical investigation of networks consisting many nodes which display highly-connected, yet realistic properties. This concatenation process is visualised in Fig.~\ref{fig:WeakReg_Explain} where it is shown how a $k=6$ regular cell can be used to generate a larger WRN. Furthermore, Fig.~\ref{fig:SepBounds}(a) depicts a number of examples of these network cells. For more precise details and discussions, see the Supplementary Material.

\begin{figure*}
\hspace{-1mm} (a) \hspace{5.45cm} (b) \hspace{5.45cm} (c)\\
\hspace{1mm}\\
\includegraphics[width=1.\linewidth]{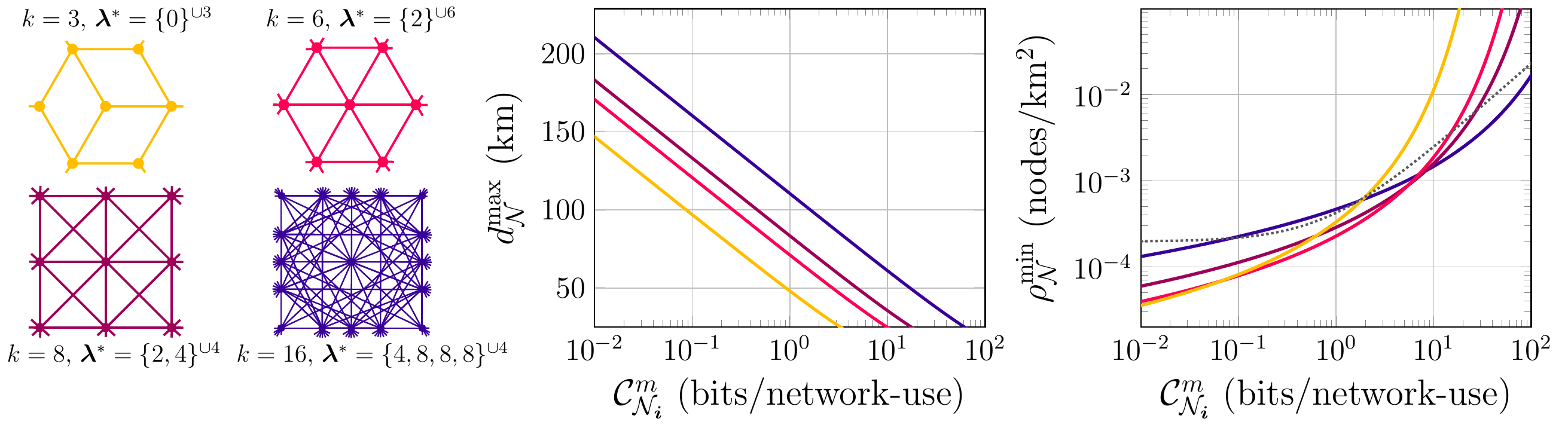}
\caption{(a) Examples of network cells that can be used to construct quantum WRNs, where $k$ denotes regularity and $\bs{\lambda}^*$ is the adjacent commonality multi-set which minimises the quantity in Eq.~(\ref{eq:CutGrowth}). These quantities characterise the network cell and larger WRNs that they can construct. 
(b) Relationship between the optimal end-to-end flooding capacity (equal to the min-neighbourhood capacity $\mc{C}_{\mc{N}_{\bs{i}}}^m$) and the maximum link-length $d_{\mc{N}}^{\max}$ required to guarantee it for bosonic lossy quantum networks according to Eq.~(\ref{eq:Dmax_MNC}). Plots are colour coordinated with the network cells. Greater network regularity leads to larger permissible ranges of channel lengths.
Panel (c) depicts this relationship between minimum nodal density $\rho_{\mc{N}}^{\min}$ with respect to optimal performance for bosonic lossy quantum networks, colour coordinated with the network cells. The grey dotted line relates the nodal density to the average flooding capacity between any pair of nodes in a Waxman Network, as in Eq.~(\ref{eq:QDense}) \cite{QuntaoRandQNets}.} 
\label{fig:SepBounds}
\end{figure*}

\subsection{Optimal Performance of Weakly-Regular Networks}

As a network becomes more highly connected, it becomes easier to locate end-to-end routes between any pair of nodes. As a result, performing network cuts requires the collection of more and more edges in a cut-set, $\tilde{C}$, in order to restrict flow along the many potential connective paths between the end-users. In a spatial network, this initiates a relationship between \text{cut-set cardinality}, $|\tilde{C}|$, and \text{distance from an end-user}. Performing cuts with edges further away from a user node requires the collection of many more edges to consolidate the partition. The further from the user nodes we begin the cut, the greater the number of potential end-to-end paths we must restrict (since we have permitted a larger flow from the user node) and thus the more edges we must collect. We call this phenomenon \textit{network cut growth}. \\

When the quality of point-to-point links in a network is consistent (link-lengths are close to the overall average), then cut-set cardinality $|\tilde{C}|$ plays a significant role in the characterisation of minimum cuts. Indeed, for dense networks with distance constrained edges, the minimum cut is often achieved by nodal isolation. This is thanks to network cut growth and consistent single-edge rates; cuts performed further away from the user-node will generate a larger multi-edge capacity since they will reliably contain more edges with similar single-edge rates. This kind of behaviour has been observed with respect to multi-path capacities in Waxman networks, and may exist in other very popular random network models \cite{WaxmanNets, QuntaoRandQNets}. This form of network cut behaviour is indicative of a well connected network and one that will achieve high rates.

This logic motivates the main theoretical tool utilised in this paper. WRNs also undergo network cut growth with respect to distance from end-users, thanks to their reliable and consistent connective properties.  Consider a large-scale $(k,\bs{\Lambda})$-WR quantum fibre-network, and a pair of end-users $\bs{i} = \{\bs{a},\bs{b}\}$ located within it. Then the cut which collects the fewest edges is that which performs user-node isolation, i.e.~collects the $k$ neighbouring edges of either end-user, generating $\mc{C}_{\mc{N}_{\bs{i}}}^m$ defined in Eq.~(\ref{eq:Nodal_Iso}). Every other cut in the network will necessarily collect more than $k$ edges in order to successfully partition the end-users. When the flooding capacity saturates Eq.~(\ref{eq:Nodal_Iso}) it is user-node isolation achieves the minimum cut. 

Unlike more complex, random network models, it is possible to analytically study network cut growth within WRNs. Here, we sketch the basic technique, and point the reader towards more sophisticated, precise arguments in the Supplementary Material. 
The basic idea is to use the quantities $k$, and $\bs{\Lambda}$  to determine how much larger a cut-set will grow when one is not permitted to cut user-neighbourhood edges. If we know how much a cut-set will grow in size with respect to distance from an end-user node, we can identify a relationship between cut-set cardinality and the link quality requirements necessary to achieve the minimum cut. Ultimately, this helps us to derive a \textit{minimum, single-edge threshold capacity} $\mc{C}_{\min}$. This reveals a minimum link-quality which when imposed upon all edges in the network will ensure that the optimal performance between end-users is guaranteed to be equal the min-user neighbourhood capacity, $\mc{C}^m(\bs{i},\mc{N}) = \mc{C}_{\mc{N}_{\bs{i}}}^m$.

This technique proves to be powerful and versatile. Indeed, given the quantum channel description of single-edges in the network, $\mc{C}_{\min}$ can be used to relate threshold properties of point-to-point quantum channels, and end-to-end performance. In the following, we employ this technique to reveal maximum link-lengths for quantum networks connected by bosonic lossy channels.

\subsection{Bosonic Lossy Quantum Networks}

When considering fibre-based networks, point-to-point links are described by bosonic pure-loss (lossy) channels. A lossy channel $\mc{L}$ with transmissivity $\eta \in (0,1)$ is a phase-insensitive Gaussian quantum channel, which transforms input quadratures $\hat{\bf{x}} = (\hat{q},\hat{p})^T$ according to $\hat{\bf{x}} \mapsto \sqrt{\eta} \hat{\bf{x}} + \sqrt{1-\eta}~\hat{\bf{x}}_{\text{env}}$ (where the environment is in a vacuum state) describing the interaction of bosonic mode with a zero-temperature bath \cite{GaussRev}. 

For lossy quantum networks, the most important property is channel length, or from a network perspective, \textit{inter-nodal separation}. For a given edge $(\bs{x},\bs{y})\in E$ connecting two users in a network, the inter-nodal separation is simply the distance $d_{\bs{xy}}$ between them. All two-way assisted quantum and private capacities of the lossy channel are precisely known via the PLOB bound \cite{PLOB},
\begin{equation}
\mc{C}_{\mc{L}}(d_{\bs{xy}}) = -\log_2\left( 1-10^{-\gamma d_{\bs{xy}} }\right),
\end{equation}
where the inter-nodal separation is related to the transmissivity via $\eta_{\bs{xy}} = 10^{-\gamma d_{\bs{xy}}}$. For current, state of the art fibre-optics the loss rate is $\gamma = 0.02$ per km (which equates to a loss rate of $0.2$ dB/km). Since these separations directly dictate the channel quality between nodes they must be precisely engineered and distributed in order to guarantee strong end-to-end performance.\par

For WR bosonic lossy fibre-networks, one can use the technique sketched in the previous section to reveal a maximum fibre-length allowed within the network to guarantee optimal end-to-end performance. Consider an end-user pair $\bs{i}=\{\bs{a},\bs{b}\}$ and a desired min-user neighbourhood capacity, $\mc{C}_{\mc{N}_{\bs{i}}}^m$. Let us define the quantities $\delta$ and $\omega$ as
\begin{align}
{\delta} \defeq \min_{\bs{\lambda}\in\bs{\Lambda}} \sum_{\lambda \in \bs{\lambda}} (k - \lambda - 1) \label{eq:CutGrowth},~~
\omega \defeq \frac{2(k-1)}{\delta},
\end{align}
which are characteristic quantities of any $(k,\bs{\Lambda})$-WR network. Then there exists a maximum fibre-length 
\begin{equation}
d_{\mc{N}}^{\max} \defeq -\frac{1}{\gamma} \log_{10}\left(1 - 2^{-\frac{1}{\delta} \mc{C}_{\mc{N}_{\bs{i}}}^m}\right),\label{eq:Dmax_MNC}
\end{equation}
for all edges in the network $(\bs{x},\bs{y})\in E$ so that the end-to-end flooding capacity is guaranteed to satisfy
\begin{equation}
 \omega\hspace{0.4mm} \mc{C}_{\mc{N}_{\bs{i}}}^m \leq \mc{C}^m(\bs{i},\mc{N}) \leq \mc{C}_{\mc{N}_{\bs{i}}}^m \label{eq:PerfBounds}.
\end{equation}
If the maximum link-length is not obeyed for all edges, $\mc{C}_{\mc{N}_{\bs{i}}}^m$ remains an upper-bound on end-to-end performance.

The parameter $\omega$ can be considered as a confidence measure on this performance guarantee. It arises out of a very worst case scenario in which all edges surrounding the neighbourhood of an end-user may be of maximum length $d_{\mc{N}}^{\max}$. This may compromise the minimum cut, potentially introducing a minor degradation of the flooding rate. 
Fortunately, $\omega$ is typically of order $10^{-1} - 1$ providing tight bounds in Eq.~(\ref{eq:PerfBounds}). For example, given the WRN cells in Fig.~\ref{fig:SepBounds}(a) we find a worst-case value of $\omega = 15/64 \approx 0.23$ for $k=16$, while for $k=3$ it is exactly $\omega =2/3$.

In general this worst-case scenario will not be true and the upper-bound in Eq.~(\ref{eq:PerfBounds}) will nearly always be saturated. Nonetheless, it is possible to provide equality by imposing a slightly stricter distance constraint on the network edges connected to end-users. That is, additionally enforcing that any edge $(\bs{x},\bs{y})\in E_{\bs{a}} \cup E_{\bs{b}}$ is shorter than
\begin{equation}
d_{\mc{N}_{\bs{i}}}^{\max} \defeq -\frac{1}{\gamma} \log_{10}\left(1 - 2^{-(\frac{1}{k-1} - \frac{1}{\delta}) \mc{C}_{\mc{N}_{\bs{i}}}^m}\right) \leq d_{\mc{N}}^{\max},\label{eq:Dmax_MNeigh}
\end{equation}
we can guarantee that the end-to-end flooding capacity satisfies
\begin{equation}
\mc{C}^m(\bs{i},\mc{N}) = \mc{C}_{\mc{N}_{\bs{i}}}^m. \label{eq:PerfEq}
\end{equation}
Throughout our investigation we focus on this more probable performance guarantee that the end-to-end flooding capacity satisfies $\mc{C}^m(\bs{i},\mc{N}) = \mc{C}_{\mc{N}_{\bs{i}}}^m$ granted that Eq.~(\ref{eq:Dmax_MNC}) is respected throughout the entire network. This allows us to effectively characterise true optimal performance for WR quantum fibre-networks.

Full details and derivations of these results can be found in the Supplementary Material. Clearly, the quantity $\delta$ in Eq.~(\ref{eq:CutGrowth}) is vital to our developments. Precisely, $\delta$ represents the smallest cut-set cardinality that can be achieved without cutting user connected edges. In more intuitive terms, it is used to monitor the network cut growth of a WRN, and allows us to derive critical quantities such as the maximum fibre-length above. Note that in Fig.~\ref{fig:SepBounds}(a) each of the WR cells are characterised by their regularity $k$ and the adjacent commonality multi-set which achieves the minimisation in $\delta$, denoted by $\bs{\lambda}^*$.
Fig.~\ref{fig:SepBounds}(b) then illustrates the relationship between flooding capacity and maximum fibre-length for a number of example WRNs. The limiting separation $d_{\mc{N}}^{\max}$ is inexorably linked with the regularity of the WRN; networks with high connectivity possess a greater tolerance for longer distance channels since the enhanced multi-path capabilities of the network outweigh the effect of poor quality channels. This is clear from the examples shown in Fig.~\ref{fig:SepBounds}, where a WRN with degree $k= 16$ can tolerate channels of $\sim60~\text{km}$ longer than one with $k= 3$. Note that the 

\subsection{Nodal Density\label{sec:NDense}}

We have discussed how WRNs can be used to describe realistic large-scale networks while maintaining analytical understanding of their optimal end-to-end performance and critical properties, e.g.~maximum link-length. We may take this analysis a step further in order to understand the relationship between end-to-end performance and network nodal density. The nodal density is defined as the number of nodes $n$ per unit area $A$ of the network, 
\begin{equation}
\rho_{\mc{N}} \defeq n/A.
\end{equation}  
Via the previous section, we may derive a maximum fibre-length $d_{\mc{N}}^{\max}$ that is necessary to guarantee some optimal end-to-end performance $\mc{C}^m(\bs{i},\mc{N}) = \mc{C}_{\mc{N}_{\bs{i}}}^m$. We may then ask the question: Is there a corresponding \textit{minimum nodal density} in which the $(k,\bs{\Lambda})$-WRN can be constructed while remaining compliant with the maximum fibre-length? It is not so easy to answer this question for completely general WR architectures. Nonetheless, for the networks studied in this work this challenge is readily tackled. 

In the Supplementary Material we show that via the concept of sparse constructions (the least dense way to construct a network given connectivity rules and link-length constraints) it is possible to derive a minimum nodal density $\rho_{\mc{N}}^{\min}$ required to achieve optimal end-to-end flooding capacity. The consistency of WRN cells in Fig.~\ref{fig:SepBounds}(a) reduce this to a geometric problem which is solvable. In summary, we can find a lower-bound on the nodal density required to achieve optimal performance $\mc{C}^m(\bs{i},\mc{N}) = \mc{C}_{\mc{N}_{\bs{i}}}^m$ given by
\begin{align}
\rho_{\mc{N}} \geq \rho_{\mc{N}}^{\min} &\defeq \xi\gamma^2 \left[\log_{10}\left(1 - 2^{-\frac{1}{\delta}\mc{C}_{\mc{N}_{\bs{i}}}^m}\right)\right]^{-2}.
\end{align}
Here, $\xi$ is a characteristic quantity of the WR network, found by studying its sparse construction. The tightness of this lower-bound depends on the manner in which the sparse construction is solved or approximated. 

Figure \ref{fig:SepBounds}(c) depicts the connection between flooding capacity and minimum nodal density for a number of types of WRNs. It is clear that there is a trade off between end-to-end performance and regular nodal degree. At low flooding rates ($10^{-2} - 10^{-1}$ bits per network use) the WR structures with lower degrees $k =3$ and $6$ demand fewer resources to achieve the same performance as those with higher degrees $k=8$ and $16$. In this regime, high degrees are not necessary everywhere in the network to achieve the flooding rates; indeed, the consistent connectivity invoked by WR designs help to maintain performance at low densities. Yet, as the flooding capacity transitions towards $1 -10$ bits per network use this behaviour changes; WRNs with low degrees demand shorter and shorter links to achieve the high rates and the inability to involve more connections at each node becomes costly. As can be seen for $k=3$ the required minimum nodal density rapidly increases, shortly followed by $k=6$ and $8$. Contrarily, the regime of high end-to-end rates is well suited to WRNs with greater regularity, $k=16$, for which the greater number of connections at each node facilitate a lower overall density. 

Simultaneously, we plot an approximation of the average flooding capacity between any pair of nodes on a Waxman network with respect to nodal density (dashed grey line) as derived in Ref.~\cite{QuntaoRandQNets}. This defines an expected flooding capacity between any pair of users, such that
\begin{equation}
\mathbb{E}_{\bs{i}} \left[\mc{C}^m (\bs{i},\mc{N}) \right] \approx \zeta( \rho_{\mc{N}} -\rho_{\text{crit}}) - 1, \label{eq:QDense}
\end{equation}
where $\rho_{\text{crit}} \approx 4.25\times 10^{-4}$, $\zeta \approx 4358$ and the average $\mathbb{E}_{\bs{i}}[\cdot]$ is taken over all possible end-user pairs in the network. We identify a kinship between the necessary $\rho_{\mc{N}}^{\min}$ predicted by WRNs and that derived for Waxman networks. As one may expect, the order and consistency of WRNs is able to promise lower resource demands at lower-rates; resulting in smaller critical nodal density predictions for the necessary density to achieve $1$ bit per network use. However, as the flooding performance increases, the flexibility of the Waxman design (its ability to utilise variable nodal degrees) renders it superior to the lower degree WR structures. In summary, there is good behavioural agreement between these models, corroborating the utility of WR structures as a valuable analytical tool for quantum network design.

\section{Comparison with Satellite Quantum Communications \label{sec:Comparison}}

\subsection{Satellite Quantum Communications}
Here, we briefly review key results which facilitate a comparison of SQC with idealised, ground-based quantum networks. For more detailed derivations and discussions of these results, please refer to Refs.~\cite{SQC,FS}.\par
Consider two users (Alice and Bob), who choose to communicate by means of an orbiting satellite (a dynamic repeater). Here we consider a ground station $G$ at approximately sea-level, and a satellite $S$ which is in orbit at an altitude $h \geq 100~\text{km}$ and variable zenith angle $\theta$. Given that the radius of the Earth is $R_{E} \approx 6371~\text{km}$, the slant distance between $G$ and $S$ is
$
z(h,\theta)  = \sqrt{h^2 + 2hR_E^2 +R_E^2 \cos^2(\theta)} - R_E\cos(\theta),
$
describing the true distance that an optical beam must travel from $G$ to/from $S$. We may consider two unique configurations for information transmission; \textit{uplink}, which refers to when $G$ is the transmitter and $S$ is the receiver, and \textit{downlink}, where the converse is true. Both configurations will identically admit the effects of free-space diffraction (beam-spot size widening) and atmospheric extinction (caused by molecular/aerosol absorption as the beam propagates). However additional loss/noise effects emerge with respect to uplink and downlink protocols, which invokes an asymmetry in their communication performance.\par
The effects of turbulence (caused by fluctuations in the atmospheric refractive index) and pointing errors (alignment of the optical signal with the receiver) are responsible for beam wandering, which instigates a fading process for the communication channels. For uplink protocols, turbulence is a significant factor for loss properties of the ground-satellite channel since it impacts the propagating beam immediately after transmission. However, pointing errors can be reduced thanks to the ability to easily access and optimise adaptive optics at ground level. In downlink these effects are reversed. Turbulence is not a factor until the beam reaches low altitudes, at which point the beam has already spread via diffraction. Hence turbulence can be neglected for downlink, but pointing errors mut be considered due to limited onboard access and resources.\par
Considering each of these physical effects characterising the lossy free-space channel, it is possible to present an ultimate limit on the secret-key capacity $K$ for SQC \cite{SQC},
\begin{equation}
K \leq -\Delta\left(\eta,\sigma\right) \log_2\left( 1-\eta \right). \label{eq:SatUltLims}
\end{equation}
Here $\Delta\left(\eta,\sigma\right)$ is a correction factor to the PLOB bound, where $\eta \defeq \eta(h,\theta)$ is an effective transmissivity which is a function of geometric position, encompassing all the effects of diffraction, extinction, and optical imperfections/inefficiencies. Meanwhile, $\sigma^2 = \sigma_{\text{turb}}^2 + \sigma_{\text{point}}^2$ is the variance of the Gaussian random walk of the beam centroid caused by beam wandering, with contributions from turbulence and/or pointing-errors.\par
This bound can be further modified to account for the presence of thermal noise, which is highly dependent upon time of day (day or night-time) and weather conditions (cloudy or clear skies). For night-time communications, background noise is practically negligible, and the above bound requires little modification. However, for day-time operations this is generally not the case and the free-space lossy channels must be described as thermal-loss channels which account for additional noise.\par

\subsection{Practical Key-Rates for Satellite Quantum Communications\label{sec:PracKeyRates}}

The bound in Eq.~(\ref{eq:SatUltLims}) is an ultimate upper-bound on the capacity of a ground-to-satellite communication channel, it is important to provide an assessment of realistic and practical protocols which embody achievable lower-bounds for SQC. These lower-bounds will facilitate comparisons with global quantum networks, and help deduce the conditions for which we can expect satellite advantage for long-distance quantum communications. 

Here we summarise some achievable rates for different satellite configurations. We consider practical, composably-secure secret key-rates achievable from the pilot-guided and post-selected CV-QKD protocol studied in Refs.~\cite{FS,SQC}. The main concept of this protocol is to encode information into Gaussian-modulated coherent states, randomly interleaved by highly energetic pilot pulses used to monitor the transmissivity and fading properties of the free-space channel in real time (facilitating the use of classical post-selection). This protocol has been comprehensively extended to account for the physical scenario of satellite quantum communications, resulting in realistic and practical rates.

We may consider the employment of such a protocol in conjunction with a near-polar sun-synchronous satellite used to communicate between two ground stations.  This type of orbit ensures a consistent fly-over time for any point on the Earth's surface, such that the satellite passes over any point at the same local mean solar time each day. This provides the possibility of stable conditions for satellite communications at around the same time each day. Let us assume that the stations lie along the orbital path such that the satellite crosses both of their zenith positions (which happens once per day). We further assume a worst-case scenario such that the stations only interact with the satellite when the zenith positions are crossed, and that both stations assume similar operational conditions.\par

It is possible to quantify the performance of satellite communications by considering a \textit{daily key rates}, i.e.~the number of secret-bits that may be shared per day. This allows us to utilise an average orbital rate $R_{\text{orb}}$ associated with up/downlink operations in day/night-time, representing an average secret-key rate per link usage. Thanks to the dynamic nature of SQC, and the fact that we consider communication with both stations only once per day, this daily rate will be constant with respect to ground based end-to-end distances. 
The number of secret-bits that can be shared in a zenith-crossing passage is then given by the effective transit time for the quantum communications $t_Q(h)$ as a function of the altitude, and a typical clock frequency which we set as $\alpha = 10~\text{MHz}$. The average daily-rate in a given configuration is thus
\begin{equation}
R_{\text{daily}}^{\text{sat}} \approx \alpha \>t_Q(h) \> R_{\text{orb}}^{i} ,\label{eq:DailyRate}
\end{equation}
for which $i$ labels the up/downlink and day/night-time.

For downlink operations at altitude $h=530~\text{km}$, initial beam-spot-size $\omega_0 = 40\text{ cm} $, receiver aperture $a_R = 1\text{ m}$, these setup parameters lead to the night-time/day-time rates \cite{SQC},
\begin{align}
&R_{\text{orb}}^{\text{down}} \approx 
\begin{cases}
3.066 \times 10^{-2}& \text{bits/use (night)},\\
3.041 \times 10^{-2}& \text{bits/use (day)}.
\end{cases}
\label{eq:Downlink}
\end{align}
For uplink, we consider and altitude $h=103~\text{km}$ and similar setups (but now with a spot-size $\omega_0 = 60\text{ cm}$ and wider aperture $a_R = 2\text{ m}$) leading to the rate,
\begin{align}
&R_{\text{orb}}^{\text{up}} \approx 
\begin{cases}
4.244 \times 10^{-2}& \text{bits/use (night)},\\
2.737 \times 10^{-2}& \text{bits/use (day)}.
\end{cases}
\label{eq:Uplink}
\end{align}

Notice that in both configurations the day and night time rates are very similar. This is thanks to  effective noise-filtering that can be performed with this kind of CV-QKD protocol. Such protocols are able to realistically exploit CV quantum systems and interferometric measurements in order to achieve much narrower frequency filters than is possible with DV protocols (see Ref.~\cite{SQC} for more details). As a result, the increased background thermal noise experienced at the receiver in day time does not significantly deteriorate the rate. 

\begin{figure*}
\includegraphics[width=\linewidth]{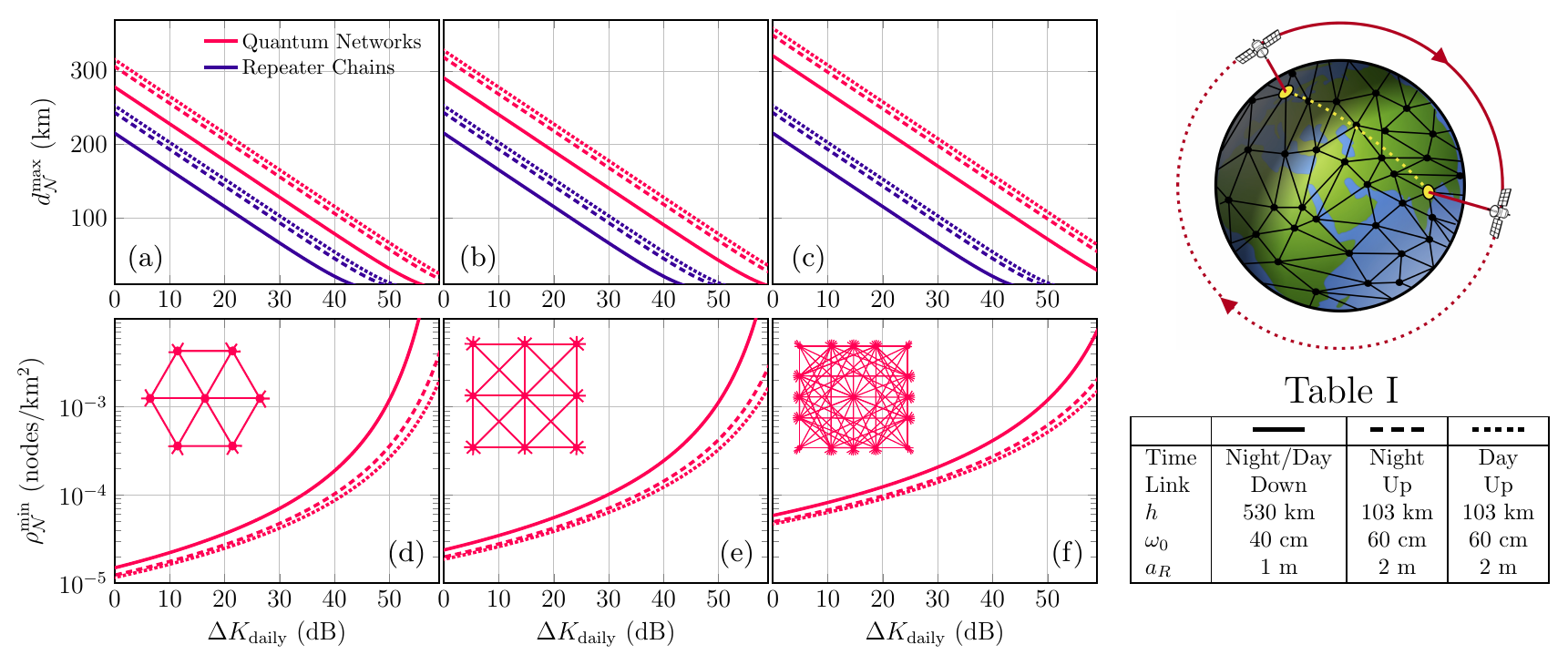}
\caption{The daily secret-key-rate advantage $\Delta K_{\text{daily}}$ in Eq.~(\ref{eq:SKAdv}) achieved by fibre-based quantum WRNs and repeater-chains with capacity achieving links over a single, sun-synchronous satellite-based repeater operating at practical, achievable rates from Eqs.~(\ref{eq:Downlink}) and (\ref{eq:Uplink}). The architecture of each WRN is shown inside of Panels (d) to (f), such that vertically aligned panels use the same architecture. Panels (a) to (c) plot $\Delta K_{\text{daily}}$ with respect to maximum fibre-length permitted within each structure, $d_{\mc{N}}^{\max}$. Panels (d) to (f) plot the relationship between the daily rate advantage $\Delta K_{\text{daily}}$ and minimum nodal density required in each WRN to achieve it. Satellite-based advantage can be achieved when $\Delta K_{\text{daily}} \leq 0$. All considered satellite setup parameters are shown in the table describing the operation time, direction, altitude $h$, spot-size $\omega_0$, receiver aperture $a_R$. }
\label{fig:SatelliteComp}
\end{figure*}

\subsection{Comparison with Ground-Based Networks}

As we have established in previous sections, end-to-end distance independence is a critical design feature for the construction of effective quantum networks. It is a feature that can be achieved, provided that one carefully monitors link-length, nodal density and the limits of quantum communication rates. Yet, as shown in the previous section, it can be very resource intensive and costly to promise strong end-to-end rates between long-distance end-users if we choose to solely utilise ground-based fibre networks. For this reason, it is important to understand the limits of large-scale quantum networks for long-range communication. Moreover, it is invaluable to determine \text{when} SQC may be superior and offer a feasible, cost-efficient route to global quantum communication.

Determination of \textit{when} SQC is advantageous requires a strict, quantitative comparison with ground-based fibre networks. In this section we aim to benchmark the optimal performance of global quantum fibre networks against practical, near-term SQC capabilities. More precisely, we compare daily secret key-rates obtained between globally distant end-users via:
\begin{itemize} 
\item[(\textit{i})] A global-scale $(k,\bs{\Lambda})$-WR fibre network with capacity achieving links.
\item[(\textit{ii})] A single, sun-synchronous satellite operating at the achievable rates in Eqs.~(\ref{eq:DailyRate})-(\ref{eq:Uplink}) using realistic devices and the practical CV-QKD protocol discussed in Section~\ref{sec:PracKeyRates}. 
\end{itemize}
Clearly, the resources accessed by an ideal $(k,\bs{\Lambda})$-WR fibre network are significantly greater than the single satellite, and a fairer comparison would be to consider a constellation of satellites; but that is the point. If a single, sun-synchronous satellite, operating at realistic rates is able to outperform a global fibre network within a meaningful resource regime, this offers clear evidence for the superiority (and necessity) of SQC for global quantum communications. Using the
tools developed throughout this paper, our comparison can be carried out expediently and analytically.

Assume two globally distant end-users, Alice and Bob. We need not consider a specific end-to-end distance, since the $(k,\bs{\Lambda})$-WRNs are end-to-end distance independent. By considering a daily key rate and the operational setup explained in Section~\ref{sec:PracKeyRates}, SQC is also end-to-end distance independent. We are left to compute the daily capacity of the WR fibre network. We consider that the fibre network operates constantly for a day using capacity achieving links with maximum link length $d_{\mc{N}}^{\max}$. Given $t_{\text{daily}} = 8.64\times 10^{4}$ s as the number of seconds in a day, and again assuming $\alpha = 10~\text{MHz}$, it can be shown the average number of secret-key bits per day satisfies
\begin{equation}
R_{\text{daily}}^{(k,\bs{\Lambda})}(d_{\mc{N}}^{\max}) \lesssim -\frac{ \alpha\>t_{\text{daily}}}{\delta} \log_2(1-10^{-\gamma d_{\mc{N}}^{\max}}), \label{eq:MP_dmaxfunc}
\end{equation}
where $\delta$ is defined in Eq.~(\ref{eq:CutGrowth}). Repeater-chains can be considered in a similar manner. The repeater-chain capacity is equal to the single-edge capacity associated with the longest inter-nodal separation in the chain. Hence, the average daily secret-key rate of a repeater-chain is \cite{End2End}
\begin{equation}
R_{\text{daily}}^{\text{chain}}(d_{\mc{N}}^{\max}) \lesssim -\alpha\>t_{\text{daily}}\> \log_2(1-10^{-\gamma d_{\mc{N}}^{\max}}). \label{eq:RC_dmaxfunc}
\end{equation}\par

In order to perform a quantitative comparison between satellite and ground-based quantum communications, we can compute the log-ratio between their daily-rates,
\begin{equation}
\Delta K_{\text{daily}} \defeq 10 \log_{10} \left( \frac{R_{\text{daily}}^{(k,\bs{\Lambda})}}{R_{\text{daily}}^{\text{sat}}} \right), \label{eq:SKAdv}
\end{equation}
which determines a \textit{daily-rate advantage} in decibels (dB). An analogous quantity can be derived for the repeater chain. By studying the daily-rate advantage as a function of maximum inter-nodal separation and nodal density, we can then determine conditions for which SQC begins to outperform the global, ground-based networks. That is, 
\begin{align}
\begin{aligned}
\Delta K_{\text{daily}} &> 0 \implies \text{Fibre-Network Advantage},\\
\Delta K_{\text{daily}} &= 0 \implies \text{Equal Performance},\\
\Delta K_{\text{daily}} &< 0 \implies \text{Satellite Advantage}.
\end{aligned}
\end{align}
Hence, there exists a critical inter-nodal separation $d_{\mc{N}}^{*}$ and a critical nodal density $\rho_{\mc{N}}^{*}$ for which ${\Delta K_{\text{daily}} = 0}$. Beyond $d_{\mc{N}}^*$ or below $\rho_{\mc{N}}^{*}$, a single, sun-synchronous satellite quantum repeater is more effective than a global fibre-network.

Fig.~\ref{fig:SatelliteComp} illustrates results for the daily-rate advantage over SQC for a repeater-chain, and a number of quantum WRNs with various connectivity properties. In particular, we compare the resource demands of SQC with $k=6,8$ and $16$ WRNs based on the network-cells shown in Fig.~\ref{fig:SepBounds}(a). Each architecture will possess its own unique critical values, defining a limiting property of the network. This comparison involves the consideration of a number of SQC operational setups and conditions which are summarised in Table I in Fig.~\ref{fig:SatelliteComp}; regarding the time of operation (night or day), physical direction of communication (uplink or downlink), satellite altitude, initial beam spot-size and receiver aperture radius. It is important to note that we can \textit{always} exploit the superior communication direction (downlink) for the purposes of QKD between end-users, thanks to the independence of physical and logical flow (as discussed in Section~\ref{sec:Prelims}). Therefore the critical properties $\rho_{\mc{N}}^*$ and $d_{\mc{N}}^*$ are computed as the values for which $\Delta K_{\text{daily}} = 0$ with respect to SQC downlink rates.

In Fig.~\ref{fig:SatelliteComp}(a)-(c) we plot the maximum tolerable fibre-length permitted in a repeater chain and each WRN required to guarantee $\Delta K_{\text{daily}}$ advantage over the single satellite repeater. 
The critical fibre-length for a quantum repeater chain operating at the ultimate limit is $d_{\text{rep}}^* \approx 215~\text{km}$, which offers a lower-bound on repeater-assisted, ground-based strategies. This can be extended by quantum networks using multi-path routing strategies, as WRNs are able to tolerate longer lossy channels at the expense of greater resource demands. 
This is clear from the results in Fig.~\ref{fig:SatelliteComp}, where extending the critical separation by approximately $100~\text{km}$ requires a $k=16$ regular network, for which $d_{\mc{N}}^{*} \approx 320$ $\text{km}$.

We may also identify the minimum required WRN nodal density, $\rho_{\mc{N}}^*$, for obtaining ground-based advantage over a single satellite, plotted in Fig.~\ref{fig:SatelliteComp}(d)-(f). Analysis of this property provides an appreciation of the resources demanded by these fibre-networks. As similarly identified in Section~\ref{sec:NDense}, while the WRNs with lower regularity are constrained to shorter link-lengths, the required nodal density at poorer end-to-end rates (low levels of advantage) is smaller than that of better connected designs. We find that the critical nodal densities $\rho_{\mc{N}}^*$ are of order $10^{-5}$ nodes per $\text{km}^{2}$, e.g.~for $k=6$ we find that $\rho_{\mc{N}}^* \approx 1.49\times 10^{-5}$, for $k=16$ we gather $\rho_{\mc{N}}^* \approx 5.84\times 10^{-5}$ etc.

These are expensive values when put into the perspective of a global communication scenario. Let us take a na\"{i}ve scenario from a practical point of view, but one that is informative nonetheless. Consider quantum communication between distant end-users located in remote cities across continental Europe (e.g.~Paris to Moscow) whose land surface area spans approximately $A \approx 1\times 10^{7} \text{ km}^2$. In terms of truly global communications this is relatively local. We can choose to communicate between remote cities using a satellite in orbit acting as a dynamic quantum repeater. Alternatively, we can construct a quantum fibre-network across the continent. In this scenario, for an ideal $k=6$ WR quantum fibre-network operating at its ultimate flooding capacity to simply match the already achievable daily-rate of a single, sun-synchronous satellite, would require at least $n\geq A \rho_{\mc{N}}^* \approx 150$ repeater stations operating constantly for 24 hours. Clearly, a network of this form operating at realistic rates, under stricter physical conditions (considering thermal noise) would demand even greater resources. 

While the classical internet can exploit fibre-optic links which are thousands of kilometres long, a fibre-based quantum internet is severely limited by short link-lengths, resulting in remarkably costly resources for tasks that are already within reach of SQC. These results strongly suggest that a future quantum internet will significantly benefit from the use of SQC, and will be integral to the construction of global quantum communication networks.

\section{Conclusion \label{sec:Concl}}
In this work, we have investigated the optimal performance of global, quantum communication networks to characterise the ultimate limits of a fibre-based quantum internet. This analysis is based on an underlying network architecture that exploits weak-regularity to construct powerful, highly-connected networks. Crucially, these bounds allow us to benchmark the performance of a global quantum network versus that of a single sun-synchronous satellite acting as a dynamic repeater. The result of this comparison emphasises the power of SQC, and vast network resources that are required to outperform a single satellite in orbit at global distances. These findings strongly motivate the utilisation of ground-satellite connections within large-scale quantum networks. It is clear that free-space ground-satellite links will be integral to long-range quantum communications, as their co-operation with ground-based infrastructure as  dynamic repeaters will be invaluable. \par
This work introduces useful, analytical techniques for the study of ideal quantum networks which can be readily employed for future investigative paths. Indeed, the study of hybrid fibre/satellite networks is a topic of immediate interest; exploiting the power of SQC to enhance (rather than compete with) ground-based networks. Furthermore, the expansion of these methods to incorporate multiple satellites introduces the possibility of highly transmissive satellite-satellite channels at high altitudes.


%

\widetext
\bigskip
\newpage
\begin{minipage}{\textwidth}
\centering
\bigskip
\large 
\bf Supplementary Material: Analytical Methods for High-Rate Global Quantum Networks
\end{minipage}
\bigskip

\thispagestyle{empty}

\bigskip
In this supplementary material we provide detailed proofs for the results presented in the main paper, and discuss in-depth some of the key mathematical tools utilised throughout this work. In Section~\ref{sec:WRNs} the formal definitions of regular graphs and weakly-regular graphs are discussed, providing greater context and elaborating upon specific definitions. In Section~\ref{sec:OptPerf} we prove the main lemmas, theorems and corollaries utilised within the text allowing us to derive single-edge threshold capacities for weakly-regular networks required to guarantee optimal performance. We then apply these theorems in the context of bosonic lossy quantum channels. Finally, Section~\ref{sec:NDs} develops the relationship between maximum link-lengths and minimum nodal density through the concept of sparse constructions, deriving a connection between optimal performance and the minimum nodal densities of a range of weakly-regular structures. 

\setcounter{section}{0}
\setcounter{equation}{0}
\setcounter{figure}{0}

\section{Weakly-Regular Networks (WRNs)\label{sec:WRNs}}

In this section we explicitly introduce the concept of weak-regularity and weakly-regular networks (WRNs) using graph theoretic concepts. We provide finer context for the purposes of WRNs studied in the main text. 

\subsection{Graphs, Neighbour Sharing and Commonality}

Consider an undirected, finite graph $\mc{N} = (P,E)$ consisting of $n$ nodes in the node set $P$, and interconnected by edges in the edge set $E$. Discussed and motivated in the main-text, such a graph underlies the description of a network such that each edge (defined by an unordered pair $(\bs{x},\bs{y})\in E$) represents a communication channel $\mc{E}_{\bs{xy}}$ between network repeaters/end-users at each node. 
The ability to perform communication on a network is characterised by (\textit{i}) the communication channels which compose the network, and (\textit{ii}) the distribution of network nodes and edges resulting in a topology. 
Here, we explicitly define some key network properties that contribute to its overall topology and ultimately its end-to-end performance. 

An essential network property is nodal degree, i.e.~the number of nodes to which a given node is connected. Defining the neighbourhood of a node $\bs{x}\in P$ as
\begin{equation}
N_{\bs{x}} \defeq \{ \bs{y} \in P ~|~ \{\bs{x},\bs{y})\in E\},
\end{equation}
then the degree of the node $\bs{x}$ is equal to the cardinality of its neighbourhood
$
\text{deg}(\bs{x}) \defeq |N_{\bs{x}}|.
$
Hence, the node $\bs{x}$ has exactly $\text{deg}(\bs{x})$ neighbours. We can also define an edge-neighbourhood of $\bs{x}$ as all the edges which connect $\bs{x}$ to its neighbours,
\begin{equation}
E_{\bs{x}} \defeq \{ (\bs{x},\bs{y}) ~|~\bs{y}\in N_{\bs{x}}\}.
\end{equation}

Nodal degree, and its distribution across a network, is hugely influential on the overall performance of an architecture. However, the degree alone does not give an indication of \textit{how} a node $\bs{x}$ is connected to all of its neighbours. One may ask; are the neighbours also highly connected to one another, or are each of the neighbours distant and disconnected? Answering these questions can be very informative, and provide significant insight into the connectivity and robustness of a network. For this reason, we define useful parameters that contribute to these features. Namely, we utilise the concept of \textit{commonality}.

Commonality is a pairwise nodal property which describes neighbour sharing between nodes. Given a pair of nodes $\bs{x},\bs{y} \in P$ the commonality defines how many neighbours that $\bs{x}$ and $\bs{y}$ \textit{have in common}. Neighbour sharing behaviour may vary significantly depending on whether $\bs{x}$ and $\bs{y}$ are already connected (adjacent) and are perhaps close by; or are disconnected (non-adjacent) and perhaps distant. 
Therefore, we provide the following pair of definitions of commonality:

\begin{defin} \emph{(Adjacent Commonality):} The number of common neighbours shared by adjacent (connected) nodes. Precisely, given that $(\bs{x},\bs{y})\in E$, the adjacent commonality between this pair of nodes is $ \lambda(\bs{x},\bs{y}) \defeq |N_{\bs{x}} \cap N_{\bs{y}}|,$
so that $\lambda(\bs{x},\bs{y})$ counts the number of common neighbours shared between the nodes $\bs{x}$ and $\bs{y}$.
\end{defin}

\begin{defin} \emph{(Non-Adjacent Commonality):} The number of common neighbours shared by non-adjacent (non-directly-connected) nodes. Precisely, given that $(\bs{x},\bs{y}) \notin E$, the non-adjacent commonality is computed by $\mu(\bs{x},\bs{y}) \defeq |N_{\bs{x}} \cap N_{\bs{y}}|$,
so that $\mu(\bs{x},\bs{y})$ counts the number of common neighbours shared between the nodes $\bs{x}$ and $\bs{y}$. 
\end{defin}

\subsection{Regular Graphs}

Let us introduce the notion of regularity. Consider an undirected, finite graph $\mc{N} = (P,E)$ of $n$-nodes. A graph is defined as $k$-regular if all nodes in the graph possess exactly the same degree $k$, i.e.~the neighbourhood of any node consists of strictly $k$ nodes, 
\begin{equation}
\text{deg}(\bs{x}) = |N_{\bs{x}}| = k, ~\forall \bs{x} \in P.
\end{equation} 
Regularity significantly simplifies the connective properties of network by assuming a consistency of nodal degree. Clearly, in realistic communication networks there exist disparities of nodal degree throughout the network, as some nodes will be highly connected and others less so. Nonetheless, understanding the ability to communicate on a regular graph can help provide important information for more realistic structures.
The class of $k$-regular graphs is very broad, and more detailed classes can be defined. 

\subsubsection{Strongly Regular Graphs}
\text{Strongly Regular} (SR) graphs satisfy strict connective properties. A graph $\mc{N} = (P,E)$ is SR if it has $n$-nodes which are $k$-regular, its commonality properties are constant
\begin{align}
\lambda(\bs{x},\bs{y}) = \lambda,~\forall \bs{x},\bs{y} \in P \text{ s.t } (\bs{x},\bs{y})\in E,\\
\mu(\bs{x},\bs{y}) = \mu,~\forall \bs{x},\bs{y} \in P \text{ s.t } (\bs{x},\bs{y})\notin E,
\end{align}
and these parameters follow the relation
\begin{equation}
\mu(n-k-1) = k(k-\lambda - 1).
\end{equation}
SR graphs may be well connected, but their architectures are very strict; satisfying all of these constraints will typically result in a network with a small number of nodes. Indeed, the parameters $k,\mu,\lambda$ inhibit the ability to use a large number of nodes rendering them impractical for network design.

\subsubsection{Weakly Regular Graphs}

\begin{figure}[t!]
\hspace{-2cm} (a) \hspace{6cm} (b) \\
\hspace{1mm}\\
\includegraphics[width=0.475\linewidth]{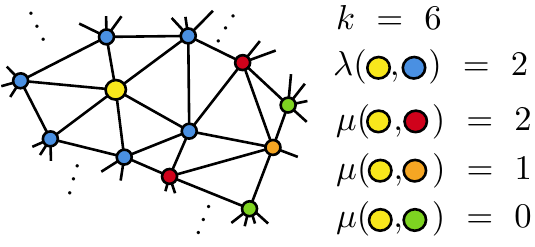}\hspace{5mm}
\includegraphics[width=0.475\linewidth]{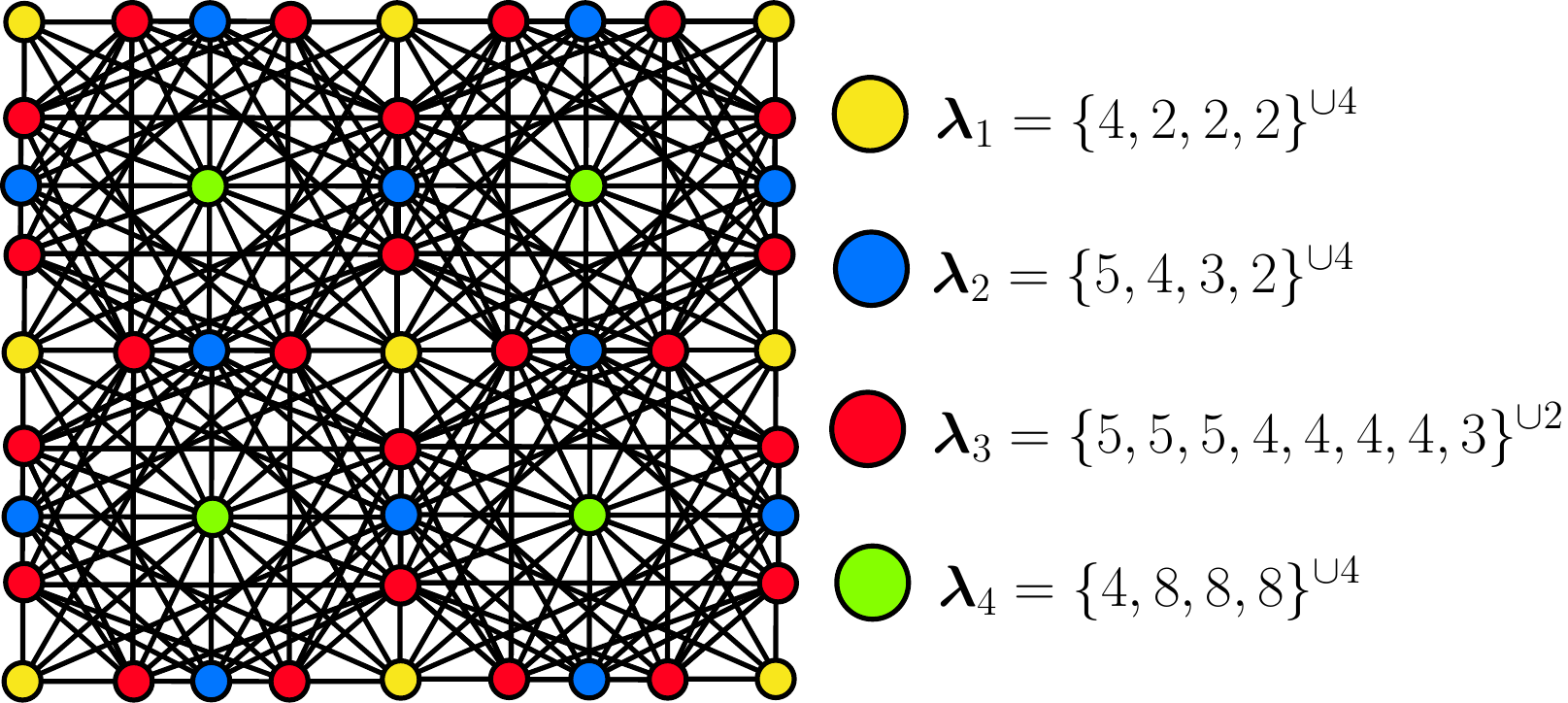}
\caption{(a) A sub-graph from a $(k,{\lambda},\bs{\mu}) = (6,2,\{0,1,2\})$-weakly regular network. Considering the yellow node as an end-user, the blue nodes thus represent the user neighbourhood, with a uniform adjacent commonality of $\lambda = 2$. The non-adjacent commonality decreases as nodes increase in distance from the end-user. (b) A $k=16$ weakly-regular network with inconsistent adjacent commonality properties. This network is scalable so that a single network cells can be concatenated to construct a larger $k=16$ internally-WR network. For any node in the network, its $\bs{\lambda}$  will be one of those from the set ${{\bs{\Lambda}} = \{ \bs{\lambda}_1, \bs{\lambda}_2, \bs{\lambda}_3, \bs{\lambda}_4\}}$. Each adjacent commonality multiset is colour coded to its corresponding node on the graph. Note that throughout this work we employ a superscript union notation to describe the repeated union of a single set, e.g.~$\{x\}^{\cup 3} = \{x\}\cup\{x\}\cup\{x\} = \{x,x,x\}$, etc.}
\label{fig:ExWReg}
\end{figure}

A more general class is that of Weakly-Regular (WR) graphs. Any regular graph that is not SR is technically WR, and can be characterised by a more general set of connectivity properties. We may invite greater generality by loosening the strict values of the adjacent/non-adjacent commonalities $\lambda$, $\mu$ for all nodes.
Instead, we may permit nodes within the network to possess different commonality values for different pairs of nodes. To this end, we define sets which contain all the potential values for the commonality properties; an adjacent commonality set and a non-adjacent commonality set respectively,
\begin{equation}
\bs{\lambda} \defeq \{\lambda_1,\ldots,\lambda_l \},~~\bs{\mu} \defeq \{\mu_1,\ldots,\mu_m\}.
\end{equation}
 These sets summarise \textit{all the non-degenerate values of $\lambda(\bs{x},\bs{y})$ or $\mu(\bs{x},\bs{y})$ that are possible on a network}. That is,
 \begin{equation}
 \lambda(\bs{x},\bs{y}) \in \bs{\lambda}, ~ \mu(\bs{x},\bs{y}) \in \bs{\mu},~\forall \bs{x},\bs{y}\in P.
\end{equation}
 There is no restriction on the number of potential values that can be contained in $\bs{\lambda}$ or $\bs{\mu}$. Indeed, every graph (regular or not) generate their own versions of these sets. 

As a simple example, we illustrate a weakly-regular sub-graph from a larger network in Fig.~\ref{fig:ExWReg}. Clearly, the degree of the network is $k=6$, and the commonality properties are also illustrated by considering adjacent and non-adjacent nodes with respect to a root node. Provided that the regularity in Fig.~\ref{fig:ExWReg} is consistent throughout the network, it can be shown that the adjacent commonality is always $\lambda = 2$ for any pair of nodes and the non-adjacent commonality can be $\mu \in \bs{\mu} = \{0,1,2\}$. 

\subsubsection{Useful Parameterisation of Weakly-Regular Graphs}

In complete generality, a graph can possess unique commonality values between any pair of nodes, leading to a vast collection of possible values in $\bs{\lambda}$ and $\bs{\mu}$, with little to no structure. However, for various architectures, such as WRNs, this will not be true and there may exist a level of consistency which simplifies their analytical treatments. 
Here, we introduce a more useful and intuitive representation of WRNs. Consider a node $\bs{x}$ on a $k$-WR graph. We can define $k$ element \textit{multiset} (a modified set which may contain multiple copies of the same element) which contains all the information about neighbour sharing between the node $\bs{x}$ and each of its $k$ neighbours,
\begin{equation}
\bs{\lambda}_{\bs{x}} \defeq \{ \lambda(\bs{x},\bs{y}_1), \ldots,\lambda(\bs{x},\bs{y}_k)\} = \{ \lambda(\bs{x},\bs{y}) ~|~ \bs{y} \in N_{\bs{x}}\}.
\end{equation}
Hence, $\bs{\lambda}_{\bs{x}}$ describes a ``local" adjacent commonality description, bespoke to the node $\bs{x}$. 

All nodes on any graph (regular or not) possess an adjacent commonality multiset $\bs{\lambda}_{\bs{x}}$ which describes neighbour sharing qualities with respect to their neighbourhoods, so there exists a unique $\bs{\lambda}_{\bs{x}}$ for all $\bs{x}\in P$. Therefore we can summarise the neighbour sharing properties of an entire network $\mc{N}$ by collecting all of the possible adjacent commonality multisets contained within it into a superset
\begin{equation}
{\bs{\Lambda}} \defeq \{ \bs{\lambda}_{\bs{x}_1}, \bs{\lambda}_{\bs{x}_2}, \ldots, \bs{\lambda}_{\bs{x}_n} \} = \{ \bs{\lambda}_{\bs{x}} ~|~{\bs{x}\in P}\}.
\end{equation}
Hence, for any node $\bs{x}$ in a network, the adjacent commonality multiset of this node can be found in ${\bs{\Lambda}}$. Note that we now define ${\bs{\Lambda}}$ as a strict set, \textit{not} a multiset. In many cases of interest, such as WRNs with high levels of symmetry, the adjacent commonalities $\bs{\lambda}_{\bs{x}}$ may be highly degenerate across the network, i.e.~many nodes possess similar neighbour sharing qualities.  Consequently, since $\bs{\Lambda}$ is a strict set it contains \textit{only} the non-degenerate $\bs{\lambda}_{\bs{x}}$ multisets. Here we state a formal definition:

\begin{defin} \emph{(Adjacent Commonality Superset):} Any graph $\mc{N} = (P,E)$ possesses an adjacent commonality superset $\bs{\Lambda} \defeq \{ \bs{\lambda}_{\bs{x}} ~|~\bs{x}\in P \}$ containing all the possible, non-degenerate adjacent commonality multisets $\bs{\lambda}_{\bs{x}}$ which describe the neighbour sharing properties of connected nodes. 
\end{defin}

For many of the architectures that we consider in this work, ${\bs{\Lambda}}$ only contains one acceptable adjacent commonality multiset, ${\bs{\Lambda}} = \{\bs{\lambda}\}$. This is evident in Fig.~\ref{fig:ExWReg}(a) where every pair of adjacent nodes always share exactly two neighbours. In general, this may occur when there are high levels of symmetry/small $k$ regularity in the network structure. 
However, this is not compulsory. Flexibility in ${\bs{\Lambda}}$ can allow us to describe more complex designs with higher nodal degrees. For instance, Fig.~\ref{fig:ExWReg}(b) depicts a $k=16$ regular network which may be a portion of a larger network. All nodes satisfy $k=16$, but there are four unique adjacent commonality multisets ${\bs{\Lambda}} = \{ \bs{\lambda}_1, \bs{\lambda}_2, \bs{\lambda}_3, \bs{\lambda}_4\}$ for any network node. As we reveal in the main-text, it is vital to understand the graphical properties of WRNs in order to properly characterise end-to-end performance for embedded end-users.

In contrast, it's not overly useful to define an analogous language for the non-adjacent commonality properties of a WR graph, $\bs{\mu}$. A node-specific non-adjacent commonality object ${\bs{\mu}}_{\bs{x}}$ would collect the number of shared neighbours between $\bs{x}$ and \textit{all nodes on the network outside of its neighbourhood}. For large-scale networks this is a potentially huge number of nodes and for the most part will not give valuable information. 
Hence, in this work we will define WR graphs according to the properties $(n,k,\bs{\Lambda}, \bs{\mu})$ along with the definitions and discussions in this section. This provides us with the most effective language to investigate this interesting graph class.\\

In the context of large-scale quantum networks, some aspects of WR architectures still need to be addressed. Namely, there are properties that are strictly defined by the parameters $n,k,\bs{\Lambda}$ and $\bs{\mu}$ which require some further discussion in order to qualify the WRN structures investigated in our work. This leads to a further sub-categorisation of weak-regularity into two key formats; \textit{genuine} and \textit{internal} weak-regularity.

\subsection{Genuine Weak-Regularity}
We define what it means for a network to be \textit{genuinely}-WR.

\begin{defin} \text{\emph{(Genuine Weak-Regularity):}} Consider a network $\mc{N} = (P,E)$ which is $(n,k,\bs{\Lambda},\bs{\mu})$-weakly-regular. This network is Genuinely-WR if there are absolutely no violations of these connectivity properties for any node $\bs{x}\in P$ within the network. 
\end{defin}

While this may seem like a trivial definition, it will become apparent in subsequent sections why it is necessary. Genuine weak-regularity can be readily satisfied but is sometimes quite restrictive. Indeed, a WRN defined within a two-dimensional spatial area may lead to some undesirable characteristics, such as extremely long edges required to satisfy regularity for all nodes; ultimately undermining the integrity of the network. 

Nonetheless, genuine weak-regularity conditions can be easily satisfied by considering closed networks embedded on a sphere (or other appropriate closed, three-dimensional objects). \text{Global} quantum networks, in which we consider a network that spans the Earth may be appropriately and ideally modelled via genuinely-WR quantum networks. This is illustrated in Fig.~\ref{fig:GenInternal_WR}(a) where a network can be defined on the surface of a three-dimensional sphere. 

\subsection{Internal Weak Regularity}

As mentioned, defining regularity conditions on a two-dimensional plane can lead to unwanted features, such as extremely long edges used to ``close" the network and satisfy all regularity conditions. Genuine weak-regularity avoids these features by considering closed networks embedded on some 3 dimensional surface. This may make sense for networks which span a planet, but for smaller areas this is not practical. 

Hence, we may provide an alternative model of network connectivity. It is possible to define a network that satisfies the WR connectivity properties \textit{within a network boundary}. That is, one can construct a WRN such that there exists a set of network nodes and edges that form a boundary 
\begin{align}
&P_{\text{bound}} = \{\bs{p}_1,\ldots,\bs{p}_m,\ldots\},\\
&E_{\text{bound}} = \{ (\bs{x},\bs{y}) \in E~|~ \bs{x},\bs{y} \in P_{\text{bound}}\},
\end{align}
within which all other nodes satisfy some form of weakly-regularity. That is, there exists a sub-network within this boundary $\mc{N}_{\text{int}} = (P_{\text{int}}, E_{\text{int}})$ according to the node and edge sets $P_{\text{int}} \defeq P \setminus P_{\text{bound}}, ~ E_{\text{int}} \defeq E \setminus E_{\text{bound}}$. The total network model $\mc{N}$ is clearly not genuinely WR since the boundary nodes $\bs{x} \in P_{\text{bound}}$ will violate the weak-regularity conditions. Nonetheless, the internal network $\mc{N}_{\text{int}}$ will satisfy these conditions, providing a useful architecture which can be readily defined over two-dimensional regions. 

\begin{defin} \text{\emph{(Internal Weak-Regularity):}} Consider a network $\mc{N} = (P,E)$. This network is defined as Internally-WR if there exists a network boundary $P_{\text{\emph{bound}}} \subset P$, $E_{\text{\emph{bound}}} \subset E$ such that the sub-network $\mc{N}_{\text{\emph{int}}} \defeq (P\setminus P_{\text{\emph{bound}}},E\setminus E_{\text{\emph{bound}}})$ satisfies a form of $(n,k,\bs{\Lambda},\bs{\mu})$ weak-regularity. 
\end{defin}

Fig.~\ref{fig:GenInternal_WR} provides a useful illustration of the difference between genuinely-WR and internally-WR networks. One of the most useful features of this network class is that they are easy to construct, and easy to scale. It is straightforward to construct a regular \textit{network cell}, which when concatenated with many other cells results in an internally-WR network. Such network cells can be seen in Fig.~2(a) of the main-text. This concatenation process makes it easy to consider large-scale WRNs with open boundaries. 

\subsection{Simplification of Notation}

\begin{figure}
\hspace{-1.cm} (a) Genuinely WR, \hspace{2.1cm} (b) Internally WR\\
\hspace{1cm}\\
\includegraphics[width=0.6\linewidth]{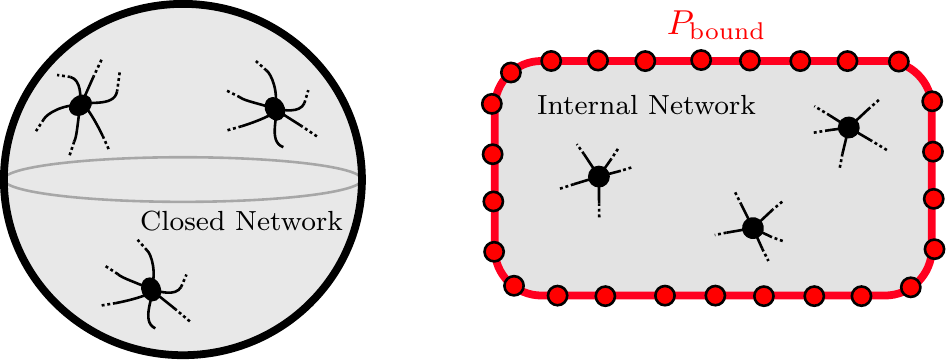}
\caption{Distinction between (a) Genuine weak-regularity and (b) Internal weak-regularity. Genuinely WR networks can be embedded on a closed, three-dimensional surface such as a sphere in order to maintain regularity and avoid boundary effects. Internally WR networks satisfy weak-regularity within some nodal boundary $P_{\text{bound}}$, allowing us to investigate open networks which are defined within some two-dimensional area. }
\label{fig:GenInternal_WR}
\end{figure}

For the purposes of our work, it is possible to simplify the notation we use to describe relevant WR structures. Since we are investigating large-scale network structures, it is not desirable to precisely define the number of nodes $n$, but allow $n$ to be encoded into other properties, e.g.~nodal density, maximum link length, etc. This can be achieved, given the crucial assumption that there are enough network nodes in so that our analysis is unaffected by boundary effects or sparsity. The need for this assumption is different depending on whether we consider internal weak-regularity or genuine weak-regularity. 

\begin{itemize}

\item \textit{Genuine weak-regularity}: By definition, we do not have any issue with boundary effects in this setting since the network is effectively closed, and all nodes are unquestionably $k$-regular. Consequently, we require a sufficient number of nodes in order to construct the closed network and ensure that there exist two end-user nodes which are not directly connected. This is not a large number of nodes and can be satisfied easily, given a particular architecture.

\item \textit{Internal weak-regularity}: In this setting, we assume that end-users nodes that we select always fall within the outer boundary of nodes, and we only consider nodes within this boundary. At the very least, we require that there are enough nodes $n$ within this boundary such that there exist two end-user nodes which are not directly connected (as this would defeat the purpose of investigating end-to-end network protocols). In general, this is a geometric packing problem specific to the weakly-regular architecture we are studying.
\end{itemize}

Henceforth, these assumptions are implicitly made within each of our WRN models. This allows us to omit the precise number of nodes $n$ from key theory throughout our work and derive general results which apply to a broad range of network structures. The number of nodes and nodal density are revisited later in our studies in order to provide adequate insight to the resource requirements of WRNs. Finally, we provide one further simplification by removing detailed reference to the possible non-adjacent-commonalities within the network, described by $\bs{\mu}$. The set $\bs{\mu}$ is important for the characterisation of short-range connective structures, detailing how many shared neighbours two non-connected nodes may possess. 
For large networks, the vast majority of non-adjacent nodes will simply share no neighbours, $\mu = 0$. This is especially true for networks which obey distance constrained connectivity rules. We find that our subsequent analyses do not require its consistent usage, therefore it can be omitted for the sake of clarity (unless otherwise specified).

Following the implicit assumptions for the necessary number of nodes required to describe a WRN and the ability to ignore the non-adjacent commonalities, we can compactly characterise a class of WR architectures via the parameters: $(k, \bs{\Lambda})$. 

\section{Optimal Performance of Weakly-Regular Networks \label{sec:OptPerf}}

The key mathematical tool we develop in this paper is the ability to accurately and analytically derive conditions for the optimal performance of quantum WRNs. This requires graph theoretic arguments and a characterisation of minimum network cuts. In this section we elucidate these arguments, allowing us to formally state and prove theory utilised in the main-text.

 \subsection{Network Cuts}

An important graph-theoretic concept for investigating network performance is that of \textit{cuts} and \textit{cut-sets}. Consider a network $\mc{N}=(P,E)$ with two remote end-users $\bs{a}, \bs{b} \in P$. An end-user pair can be represented as a set of two unique user nodes, $\bs{i} = \{ \bs{a},\bs{b} \}$. This allows us to simplify notation in many settings. 
We define a cut $C$ as a bipartition of all network nodes $P$ into two disjoint subsets of nodes $(P_{\bs{a}}, P_{\bs{b}})$ such that the end-users become completely disconnected, $\bs{a} \in P_{\bs{b}}$ and $\bs{b} \in P_{\bs{b}}$, where $P_{\bs{a}}\cap P_{\bs{b}} =\varnothing$. A cut $C$ generates an associated cut-set; a collection of network edges $\tilde{C}$ which when removed cause the partitioning. Precisely, a cut-set is defined by
\begin{equation}
\tilde{C} = \{ (\bs{x},\bs{y}) \in E~|~\bs{a} \in P_{\bs{a}}, \bs{b}\in P_{\bs{b}}\}.
\end{equation}
Under the action of a cut, a network is successfully partitioned
\begin{equation}
\mc{N} = (P,E) \xrightarrow{\text{Cut: $C$}} (P,E \setminus \tilde{C}) = ( P_{\bs{a}}\cup P_{\bs{b}}, E \setminus \tilde{C} ),
\end{equation}
so that there no longer exists a path between $\bs{a}$ and $\bs{b}$. Network cuts play a key role in the derivation of end-to-end network rates. Many network optimisation tasks can be reduced to an optimisation over all cuts with respect to single-edge/multi-edge properties.

As discussed in the main text, any valid network cut can be associated with a multi-edge capacity $\mc{C}^m(C)$ calculated by the sum of all the single-edge capacities in the cut-set,
\begin{equation}
\mc{C}^m(C) = \sum_{(\bs{x},\bs{y})\in \tilde{C}} \mc{C}_{\bs{xy}},
\end{equation} 
where $\mc{C}_{\bs{xy}} = \mc{C}(\mc{E}_{\bs{xy}})$ is the single-edge capacity associated with the channel between nodes $\bs{x}$, $\bs{y}$. 
A flooding capacity is given by the network-cut between the user-pair which minimises this multi-edge capacity,
\begin{equation}
\mc{C}^m(\bs{i},\mc{N}) = \min_C~ \mc{C}^m(C) = \min_{C} \sum_{(\bs{x},\bs{y})\in\tilde{C}} \mc{C}_{\bs{xy}}. \label{eq:flood}
\end{equation}

\subsection{Motivation}

As discussed in the main-text, solving Eq.~(\ref{eq:flood}) for a general network and capacity distribution requires a numerical treatment via the max-flow min-cut theorem. However, for any network we can always identify at least one valid cut via \textit{user-node isolation}, i.e.~cutting all the edges in the neighbourhood of one of the end-user nodes, $E_{\bs{a}}$ or $E_{\bs{b}}$. This cut totally disconnects an end-user node from the network, resulting in a successful partition. We call the multi-edge capacity associated with this kind of cut as the \textit{min-neighbourhood capacity},
\begin{equation}
\mc{C}^m(\bs{i}, \mc{N}) \leq \mc{C}_{\mc{N}_{\bs{i}}}^m \defeq \min_{\bs{j} \in \{\bs{a},\bs{b}\}} \sum_{(\bs{x},\bs{y})\in E_{\bs{j}}} \mc{C}_{\bs{xy}}.
\end{equation}

The min-neighbourhood capacity is always at least an upper-bound on the end-to-end flooding capacity. It is a strong indicator of a well connected and thus high-performance network.
Networks which are highly connected contain many end-to-end routes between any pair of network nodes. The greater the number of end-to-end routes between an end-user pair, the more difficult it is to partition them via a cut-set, i.e.~it requires more and more edges to disconnect them. As discussed in the main-text, this initiates a relationship between \text{cut-set cardinality}, $|\tilde{C}|$, and \text{distance from an end-user}. Performing cuts with edges further away from a user node requires the collection of many more edges to consolidate the partition. The further from the user nodes we begin the cut, the greater the number of potential end-to-end paths we must restrict (since we have permitted a larger flow from the user node) and thus the more edges we must collect. We call this phenomenon \textit{network cut growth}.
Once again, for general architectures and topologies it is extremely difficult to investigate the concept of network cut growth and would require numerical treatment. However, WRNs are analytically friendly and an ideal candidate for studying this concept. 

Our approach is based on the distinction between two kinds of network cuts; user-node isolation, and \text{network-bulk cuts}. Let us formally define the notion of a network-bulk:

\begin{defin}
Consider a network $\mc{N} = (P,E)$ and an end-user pair $\bs{i} = \{ \bs{a},\bs{b}\}$ who wish to communicate. We define a network-bulk with respect to this end-user pair as the sub-network $\mc{N}^{\prime} = (P^{\prime}, E^{\prime})$ which contains all the edges and nodes which are not directly connected to the end-user nodes. That is, the node and edge sets satisfy,
\begin{align}
&P^{\prime} \defeq \left\{ \bs{x} ~|~\bs{x} \in P \setminus \{\bs{a},\bs{b}\} \right\}, ~~~E^{\prime} \defeq \left\{ (\bs{x},\bs{y}) ~|~\bs{x} \in P \setminus (E_{\bs{a}} \cup E_{\bs{b}}) \right\}.
\end{align}
\end{defin}

In a large-scale network the network-bulk $\mc{N}^{\prime}$ constitutes the majority of the architecture. A \textit{network-bulk cut} can then be considered as a network-cut $C^{\prime}$ which is performed by exclusively collecting edges from the network-bulk rather than the user-neighbourhoods. By our previous arguments, when a network is well connected, cuts performed further away from either end-user refer to collections of edges \text{in the network bulk}. 

This leads to the primary motivation of our work: Via the intuition of network cut growth, we wish to derive a relationship between WR networks, user-node isolation and network-bulk cuts. We wish to show that in highly-connected architectures (such as WRNs), cut growth causes cuts in the network-bulk to be unlikely candidates for the minimum cut. As a result, this allows us to identify conditions for which the upper-bound in Eq.~(\ref{eq:flood}) is saturated and elucidate network properties for which optimal performance is guaranteed.

\subsection{Network-Bulk Cuts}

In this section, we derive some useful lemmas which help us to understand network cut growth and network-bulk cuts.

\begin{lemma}
Select two nodes on a genuinely-WR quantum network $\mc{N} = (P,E)$ that represent end-users, $\bs{a}, \bs{b} \in P$, and demand they that they do not share an edge or neighbour. The cut-set $\tilde{C}$ which contains the fewest number of edges collects $k$ edges. 
\label{lemma:IsoCut}
\end{lemma}
\begin{proof}
Menger's theorem states that for a finite, undirected graph the size of the minimum cut-set is equal to the maximum number of disjoint paths that can be found between any pair of vertices \cite{Menger1927,Aharoni2009}. Here, we are considering a $(k,\bs{\Lambda})$ weakly-regular graph with enough nodes to locate a pair of end-users which do not share an edge or neighbour. Every disjoint path will have to use one of the edges from the neighbourhood of an end-user, $N_{\bs{a}}$ and $N_{\bs{b}}$. After $k$ disjoint paths, all the edges in the neighbourhoods of the end-user nodes will have already been used by one of these paths. Consequently, no more disjoint paths can be found, as the end-users can find no route to the network-bulk. Hence, the smallest cut-set cardinality will always equal $k$.  \end{proof}

\begin{lemma}
Consider a $(k,\bs{\Lambda})$-genuinely-WR quantum network $\mc{N} = (P,E)$ such that $\bs{\Lambda} = \{\bs{\lambda}\}$. Select two nodes that represent end-users, $\bs{a}, \bs{b} \in P$, and demand that they do not share an edge or neighbour. For any cut-set $\tilde{C}$ that is restricted to edges in the network-bulk $e\in E^{\prime}$, 
\begin{equation}
\text{\emph{if} } \sum_{j=1}^k \lambda_j \leq k(k-2) \implies |\tilde{C}| \geq k.
\end{equation}
If $\lambda_j = \lambda, \forall j$ then the condition holds if $\lambda \leq k-2$.
\label{lemma:BodyCut}
\end{lemma}

\begin{proof}
For a genuine $(k,\{\bs{\lambda}\})$-regular network there will always exist a cut-set with cardinality $|\tilde{C}_{\text{iso}}| = k$, achieved by isolating the neighbourhoods of either of the end-user nodes. By Lemma 1, we also know that this is the minimum cut-set cardinality.
Meanwhile, a network cut which is limited to collecting edges on the network-bulk is unable to directly disconnect the neighbourhoods of $\bs{a}$ or $\bs{b}$ ($N_{\bs{a}}$ and $N_{\bs{b}}$ respectively). Hence, any cut which is performed on the network-bulk has to restrict flow from not just the end-users, but each of its neighbours. That is, any alternative cut-set will have to cut the unique edges in the neighbourhoods of the $\bs{a}$/$\bs{b}$'s neighbouring nodes. 

Let us consider the cut which restricts flow from all of the neighbouring nodes of either end-user. This cut-set will be either of the following:
\begin{align}
\tilde{C}_{\bs{a}} &= \bigcup_{\bs{x} \in N_{\bs{a}}}  \{ (\bs{x},\bs{y})~|~ \bs{y} \in N_{\bs{x}}\setminus \left( N_{\bs{a}} \cup N_{\bs{b}} \cup \{\bs{a},\bs{b}\}\right) \},\\
\tilde{C}_{\bs{b}} &= \bigcup_{\bs{x} \in N_{\bs{b}}}  \{ (\bs{x},\bs{y})~|~ \bs{y} \in N_{\bs{x}}\setminus \left( N_{\bs{a}} \cup N_{\bs{b}} \cup \{\bs{a},\bs{b}\}\right) \}.
\end{align}
What are the cardinalities of these cut-sets? Thanks to network regularity this is easy to derive. Our goal is to restrict flow from each of the neighbours of $\bs{a}$ or $\bs{b}$. By weak regularity, these neighbours will possess $k$ edges; they will use one edge to connect directly to $\bs{a}$ or $\bs{b}$, and $\lambda_j$ nodes will be connected to \textit{other} neighbours of $\bs{a}$ or $\bs{b}$ (by definition of adjacent commonality). As a result there will only be $(k-\lambda_j -1)$ effective edges that permit logical flow \textit{outside} of the end-user neighbour (to the rest of the network). 
By summing over all of the neighbours in either neighbourhood each of the new cut-set cardinalities are then 
\begin{equation}
|\tilde{C}_{\bs{a}/\bs{b}}| = \sum_{j=1}^k (k-\lambda_j-1). \label{eq:BulkCut}
\end{equation}

When will this quantity be greater than $|\tilde{C}_{\text{iso}}|$? This is easy to determine and retrieves the condition stated in the lemma,
\begin{equation}
\sum_{j=1}^k (k-\lambda_j-1) \geq k \implies \sum_{j=1}^k \lambda_j \leq k(k-2).
\end{equation}
If this condition holds, then we can always write $|\tilde{C}_{\bs{a}/\bs{b}}| \geq k$ as required. Any cut which is not $\tilde{C}_{\text{iso}}$ or $\tilde{C}_{\bs{a}/\bs{b}}$ will necessarily permit flow into wider parts of the network. This will always increase the number of disjoint paths from $\bs{a}$ to $\bs{b}$ within the network, since the network is weakly-regular and connectivity properties are consistent throughout the network. It will therefore have a larger cut-set cardinality. 
\end{proof}\\

Lemma~\ref{lemma:BodyCut} serves as a critical tool in our analyses. It relates analytical properties of a WRN with properties of cuts performed on its network-bulk. 
In abstract terms, it posits weak-regularity conditions for which we can be certain that a WRN undergoes network cut growth. 

While Lemma~\ref{lemma:BodyCut} has been formalised in the context of a WRN with consistent adjacent commonality properties (i.e.~$\bs{\Lambda}=\{\bs{\lambda}\}$ has only one multiset) it can be easily extended to account for multiple possible multisets. For this reason we propose the following definition.

\begin{defin} \emph{(Minimum Adjacent Commonality):} Given a $(k,\bs{\Lambda})$-WR graph, the minimum adjacent commonality multiset $\bs{\lambda}^{*} \in \bs{\Lambda}$ is that which collects the fewest edges on a network-bulk cut,
\begin{equation}
\bs{\lambda}^* = \argmin_{\bs{\lambda}\in\bs{\Lambda}} \sum_{\lambda \in \bs{\lambda}} (\lambda - k - 1).
\end{equation}
\end{defin}

The minimum adjacent commonality multiset is a characteristic of any WRN. It identifies the network nodes which possess the smallest network-bulk cuts. Ultimately, if Lemma~\ref{lemma:BodyCut} is satisfied for the minimum adjacent commonality multiset $\bs{\lambda}^*$, then it holds for all possible nodes on the network. 

\begin{lemma}
Select two nodes on a genuinely $(k,\bs{\Lambda})$-WR quantum network $\mc{N} = (P,E)$ that represent end-users, $\bs{a}, \bs{b} \in P$, and demand they that they do not share an edge or neighbour. For any cut-set $\tilde{C}$ that is restricted to edges in the network-bulk $e\in E^{\prime}$, if $\sum_{\lambda\in\bs{\lambda}^*} \lambda \leq k(k-2)$ it follows that $|\tilde{C}| \geq k.$
\label{lemma:BodyCut_II}
\end{lemma}

\begin{proof}
Since there is now variation in $\bs{\lambda}$ from node to node, the network-bulk cut that is associated with $\bs{\lambda}^*$ collects the least number of edges (by definition). Then any other node with a different $\bs{\lambda}\in \bs{\Lambda}$ must necessarily collect more edges than this. Given this consideration, the proof then follows directly from Lemma~\ref{lemma:BodyCut}. 
\end{proof}

\subsection{Network-Bulk Cuts on Internally-WRNs}

\begin{proposition}
Select two nodes on an internally-WR quantum network $\mc{N} = (P,E)$ that represent end-users, $\bs{a}, \bs{b} \in P$, and demand they that they do not share an edge.
There exists some minimum number of network nodes $n_{\min}$ for which the the results of \emph{Lemma \ref{lemma:BodyCut}/\ref{lemma:BodyCut_II}} applies to $\mc{N}$.
\label{lemma:IWRBodyCut}
\end{proposition}

This proposition is well motivated, and can be proven for a number of different WR architectures. Open boundary edges add the complication of a potential cut $C$ that utilises the boundary to find a smaller cut-set than that used in Lemma \ref{lemma:BodyCut}. However, it is always possible to construct a sufficiently large network so that a pair of end-user nodes can always be found for which Lemma \ref{lemma:BodyCut} is satisfied. We describe these end-user nodes as \textit{deeply-embedded}. 
It is possible to provide a general characterisation of $n_{\min}$ by identifying the minimum number of nodes for which there exists a cut-set $\tilde{C}^{\prime\prime}$ containing boundary edges $e\in E_{\text{bound}}$ that has a smaller cardinality than the network-bulk cut in Eq.~(\ref{eq:BulkCut}), i.e.~$|\tilde{C}^{\prime\prime}| < |\tilde{C}_{\bs{a}/\bs{b}}|$. Indeed this can be achieved, but such generality is not particularly useful in this paper, and we leave it to future works.

For now, we focus on WRN structures for which determining $n_{\min}$ is a basic geometric problem. The quantity $n_{\min}$ is the minimum number of nodes required to locate two end-users which do not share an edge or a neighbour, so that Lemma~\ref{lemma:BodyCut} is not compromised by the network boundary.
In Fig.~\ref{fig:min_nodes_check} we provide visual proofs of the minimum number of nodes required to satisfy internal-WR for the structures utilised in this paper. In each case the red regions of the network describes boundary region, while the white region resembles the internal network which is WR for the end-user nodes (which are coloured yellow). The blue dotted line is the network-bulk cut which collects the number of edges described in Lemma~\ref{lemma:BodyCut}. The green cut is the smallest cut that exploits the boundary edges in order to reduce the cut-set size. The removal of any node on each network will give rise to a smaller cut than the blue cut by exploiting the boundary edges. 

\begin{figure}[t!]
\includegraphics[width=0.775\linewidth]{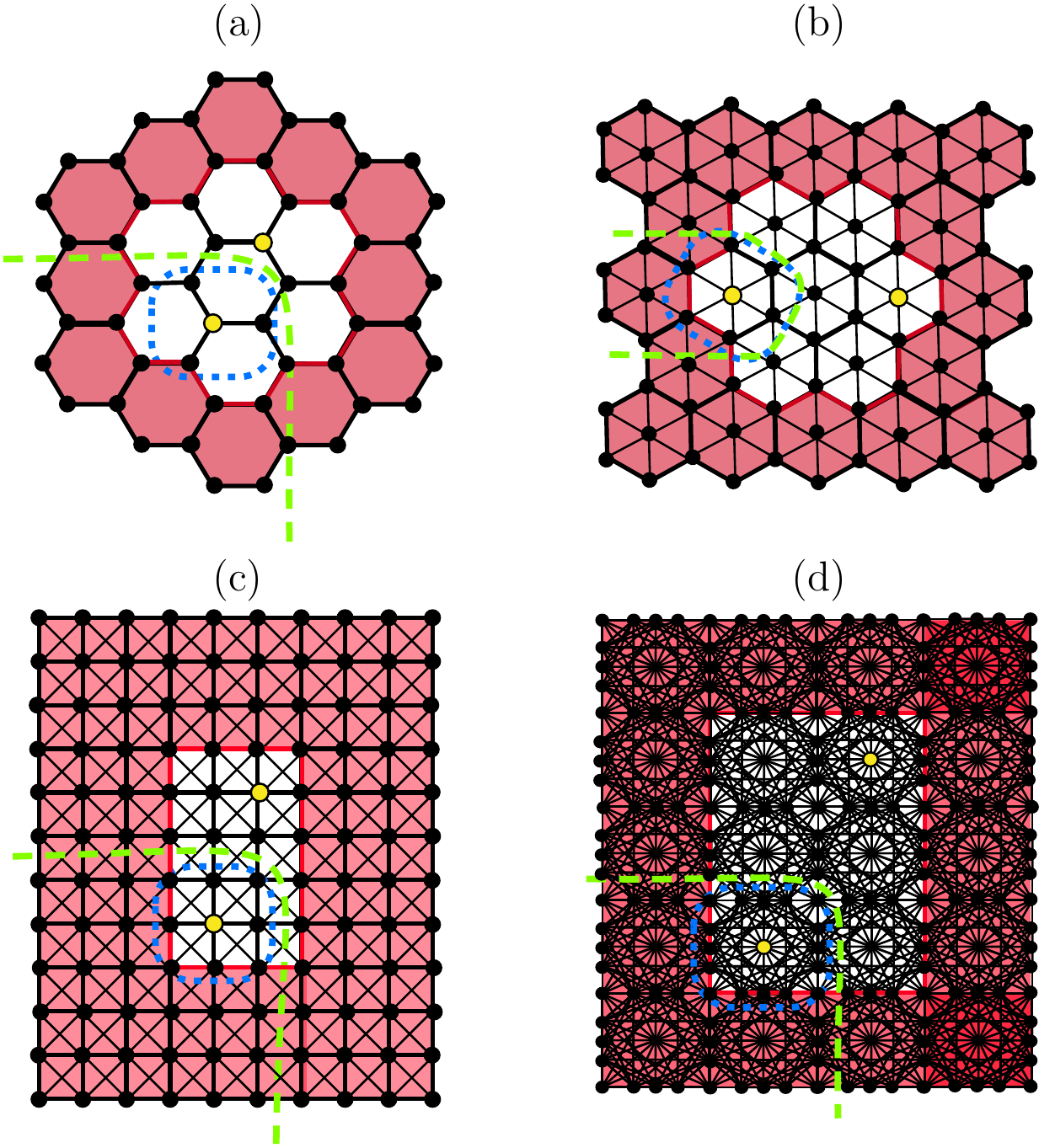}
\caption{Minimum node WRNs for a strict satisfaction of internal regularity for (a) honeycomb network, (b) hexagonal network, (c) Manhattan $k=8$ network and (d) Manhattan $k=16$ network. Each case resembles the smallest WRN for which there exist a pair of end-user nodes which do not share an edge or a neighbour, and possess minimum cardinality network-bulk cuts which are unaffected by the open boundary.}
\label{fig:min_nodes_check}
\end{figure}

As an example, consider Fig.~\ref{fig:min_nodes_check}(c). The network-bulk cut according to Lemma~\ref{lemma:BodyCut} for genuinely-WRNs will collect 32 edges (illustrated by the blue dotted line). However, the existence of the open boundary allows for different cuts which may compromise this result. The green dotted line depicts the smallest cut that would not be available on a genuinely-WR version of this network. It collects 33 edges. There does not exist any other network-bulk cut that can collect fewer edges. The minimum number of nodes required in each case:
\begin{align}
\begin{aligned}
&k=3, ~\bs{\lambda}^* = \{0\}^{\cup 3} \rightarrow n_{\min} = 54,\\
&k=6, ~\bs{\lambda}^* = \{2\}^{\cup 6} \rightarrow n_{\min} = 89,\\
&k=8, ~\bs{\lambda}^* = \{2,4\}^{\cup 4} \rightarrow n_{\min} = 110,\\
&k=16, ~\bs{\lambda}^* = \{4,8,8,8\}^{\cup 4} \rightarrow n_{\min} = 197.
\end{aligned}
\end{align}
In general, we are interested in large-scale networks with many more nodes than any of these values $n \gg n_{\min}$, so clearly the investigation of internally-WR graphs is well justified. When the neighbour sharing condition is relaxed for the end-users, this minimum number of nodes is reduced so that these constructions remain sufficient.

\subsection{Threshold Capacities}

With these lemmas in hand, we can present the key mathematical tools used throughout this paper and derive the following \textit{threshold theorems}.

\begin{theorem}
Consider a $(k,\bs{\Lambda})$-WR quantum network. Select an end-user pair $\bs{i} = \{\bs{a},\bs{b}\}$, and demand they are sufficiently distant such that they do not share an edge or neighbour. Then there exists a threshold single-edge capacity $\mc{C}_{\min}$ in the network, given by
\begin{gather}
\mc{C}_{\min} \defeq \frac{1}{\delta}\>\mc{C}_{\mc{N}_{\bs{i}}}^m, \label{eq:Thresh}
\end{gather}
where $\delta$ is a characteristic property of the network, 
$
{\delta} \defeq {\sum_{\lambda\in\bs{\lambda}^*} k - \lambda - 1},
$
such that if all single-edge capacities in the network satisfy this minimum threshold,
$\mc{C}_{\bs{xy}} \geq \mc{C}_{\min},\forall\>(\bs{x},\bs{y})\in E$
then flooding capacity is guaranteed to satisfy
\begin{equation}
\frac{2(k-1)}{\delta} \mc{C}_{\mc{N}_{\bs{i}}}^{m}\leq \mc{C}^m(\bs{i},\mc{N}) \leq \mc{C}_{\mc{N}_{\bs{i}}}^{m}. \label{eq:PerfBounds}
\end{equation}
\label{theorem:Thresh1}
\end{theorem}

\begin{proof} Let us denote the $(k,\bs{\Lambda})$-WRN $\mc{N} = (P,E)$. The network possesses a large set of valid cuts, $\textsf{C}_{\mc{N}} = \{ C_j \}_j$, which collects all of the valid network cuts $C_j$ that can successfully partition the pair of end-users. We can simultaneously define a set of \textit{cut-set cardinalities}, i.e.~if there exist $M$ valid cuts, this is a $M$-element multi-set that counts the number of edges contained in each of the valid network-cuts. More precisely, we can define this multiset
\begin{equation}
c_{\mc{N}} = \{ |\tilde{C}| ~|~ \tilde{C} \text{ s.t } C\in \textsf{C}_{\mc{N}} \}.
\end{equation}
By Lemma 1, the minimum-cut-set cardinality for the WRN $\mc{N}$ is simply equal to its regularity, i.e.~$\min (c_{\mc{N}}) =  k$, and can be achieved by isolating an end-user (cutting the edges within an end-user neighbourhood). Performing user-node isolation we simply collect the edges from the user-neighbourhood $\tilde{C} = {E_{\bs{i}}}$ to generate the min-neighbourhood capacity,
\begin{align}
\mc{C}_{\mc{N}_{\bs{i}}}^m &=  \min_{\bs{j}\in \{\bs{a},\bs{b}\}} \hspace{-1mm} \sum_{(\bs{x},\bs{y})\in {E_{\bs{j}}}} \mc{C}_{\bs{xy}}. \label{eq:MinUserMPC}
\end{align}

Now let us consider any network-bulk cut ${C}^{\prime}$ and its corresponding cut-set $\tilde{C}^{\prime}$ which is restricted to collecting edges on the network-bulk $\mc{N}^{\prime}$. These types of cuts cannot use edges from the end-user-neighbourhoods and will provide a multi-edge capacity,
\begin{equation}
{\mc{C}^m}(C^{\prime}) =  \sum_{(\bs{x},\bs{y})\in \tilde{C}^\prime} \mc{C}_{\bs{xy}}. \label{eq:NetBodMPC}
\end{equation}
 In order to ensure $\mc{C}_{\mc{N}_{\bs{i}}}^m$ is indeed the flooding capacity of the entire network, we must ensure that the minimum network-bulk based cut is \textit{never} a minimum-cut, so that ${\mc{C}^m}(C^{\prime}) \geq \mc{C}_{\mc{N}_{\bs{i}}}^m$.
 
When restricted to performing cuts only on the network-bulk, the set of possible cuts will be different from ${\textsf{C}}_{\mc{N}}$, since now certain cuts are inaccessible. Instead, we may define a new set of network-cuts $\textsf{C}_{\mc{N}^{\prime}} $ which are restricted to the network-bulk. This generates an analogous set of cut-set cardinalities 
\begin{equation}
c_{\mc{N}^{\prime}} = \{ |\tilde{C}| ~|~ \tilde{C} \text{ s.t } C\in \textsf{C}_{\mc{N}^{\prime}} \}.
\end{equation}
Using Lemma 1 we can determine the smallest network-bulk based cut on a $(k,\bs{\Lambda})$-WR network (with no boundary effects). Since $\bs{\Lambda} = \{\bs{\lambda}_{\bs{x}}~|~\bs{x}\in P\}$ may contain many different adjacent commonalities, it is always possible to lower-bound the cardinality of the smallest network-bulk cut-set by using the minimum adjacent commonality multiset $\bs{\lambda}^*$. Then we can write,
\begin{equation}
 \min(c_{\mc{N}^{\prime}}) \geq \sum_{\lambda\in\bs{\lambda}^*} (k-\lambda - 1) = \delta. 
\label{eq:delta_defin}
\end{equation}
This corresponds to the minimum number of edges that must be cut from the neighbours \textit{of the neighbours} of the minimum end-user (e.g.~the green cut-set in Fig.~\ref{fig:Cuts}).
This generates $\tilde{C}^{\prime}$ as the cut-set restricted to the network-bulk with minimum cardinality. 

In order to ensure that the flooding capacity is equal to the min-neighbourhood capacity, we want to make sure that this network-bulk cut never generates a multi-edge capacity smaller than $\mc{C}_{\mc{N}_{\bs{i}}}^m$. That is, we wish to ensure that
\begin{equation}
\min_{C^\prime \in \textsf{C}_{\mc{N}^{\prime}}} \mc{C}^m(C^{\prime}) \geq   \mc{C}_{\mc{N}_{\bs{i}}}^m. \label{eq:NB_Cond}
\end{equation}
Minimising ${\mc{C}^m}(C^{\prime})$ is achieved by setting each edge in the network-bulk to its minimum value $\mc{C}_{\min}$ and performing the cut which collects the fewest number of edges, such that
\begin{align}
 \min(c_{\mc{N}^{\prime}})\cdot\hspace{-2mm} \min_{(\bs{x},\bs{y}) \in \tilde{C}^{\prime}}  \mc{C}_{\bs{xy}}  =  \min(c_{\mc{N}^{\prime}}) \cdot \mc{C}_{\min}\geq   \mc{C}_{\mc{N}_{\bs{i}}}^m.
\end{align}
Subsequently we can derive a minimum threshold  capacity for any edge in the network,
\begin{align}
\mc{C}_{\bs{xy}} \geq \mc{C}_{\min} = \frac{ \mc{C}_{\mc{N}_{\bs{i}}}^m}{\sum_{\lambda\in\bs{\lambda}^*} (k-\lambda-1)} = \frac{1}{\delta}\> \mc{C}_{\mc{N}_{\bs{i}}}^m, \label{eq:NetProps} ~~\forall (\bs{x},\bs{y})\in E,
\end{align}
which ensures that Eq.~(\ref{eq:NB_Cond}) is always upheld. Imposing this threshold constraint ensures that any cut restricted to the network-bulk will generate a multi-edge capacity that is greater than or equal to the min-neighbourhood capacity. As a result, no cut performed exclusively on the network-bulk can ever undermine the flooding capacity.

There is now only one issue; we must identify if there exists any possible \textit{hybrid} cut that might undermine the flooding capacity being equal to the min-neighbourhood capacity. That is, is there a cut that can collect a mixture of edges contained in the user-neighbourhood \textit{and} the network-bulk? Unfortunately there is, and it must be considered. Let us take a worst-case scenario where \textit{all} of the edges in a network-bulk are of minimum threshold capacity $\mc{C}_{\min}$. Furthermore, let's consider that the min-neighbourhood capacity $\mc{C}_{\mc{N}_{\bs{i}}}^m$ is generated by a user-neighbourhood in which has $(k-1)$ edges of capacity $\mc{C}_{\min}$ and one edge with capacity $\mc{C}_{\mc{N}_{\bs{i}}}^m - (k-1)\mc{C}_{\min}$. That is,
\begin{align}
\mc{C}_{\mc{N}_{\bs{i}}}^m &= (k-1) \mc{C}_{\min}  + \left[ \mc{C}_{\mc{N}_{\bs{i}}}^m - (k-1)\mc{C}_{\min}\right].
\end{align}
This is a worst-case situation in which one neighbourhood edge contains the majority of the min-neighbourhood capacity. In this scenario, it is possible to cut the $(k-1)$ edges in the neighbourhood which have the threshold value, and then to cut and additional $(k- 1)$ edges in the network-bulk which are connected to the largest capacity edge in the neighbourhood \textit{instead} of this neighbourhood edge. This results in a hybrid cut $C{''}$ which generates a multi-edge capacity
\begin{equation}
\mc{C}^m(C^{\prime\prime}) \geq 2(k-1)\mc{C}_{\min} = \frac{2(k-1)}{\delta} \mc{C}_{\mc{N}_{\bs{i}}}^m.
\end{equation}
This is an absolute worst-case scenario for the network design, placing a lower-bound on the end-to-end flooding capacity.

Consequently, provided that $\mc{C}_{\bs{xy}} \geq \mc{C}_{\min},~\forall (\bs{x},\bs{y}) \in E$ then the flooding capacity always satisfies
\begin{equation}
\frac{2(k-1)}{\delta} \mc{C}_{\mc{N}_{\bs{i}}}^{m}\leq \mc{C}^m(\bs{i},\mc{N}) \leq \mc{C}_{\mc{N}_{\bs{i}}}^{m}.
\end{equation}
as required. This reveals a single-edge threshold condition for all edges in the network so to ensure that end-to-end performance is guaranteed within tight bounds.  \end{proof}\\

Theorem~\ref{theorem:Thresh1} therefore allows us to place tight performance bounds on the flooding capacity of a quantum WRN. Using only the connectivity properties of the architecture itself, and a desired end-to-end performance, we can identify a single-edge capacity constraint for all network edges. This is extremely useful, and a key result in this work.

It is also possible to identify what additional constraints are necessary to not just guarantee a tight window of performance, but guarantee exact, optimal performance. This is achieved in the following theorem.

\begin{figure}
\includegraphics[width=0.5\linewidth]{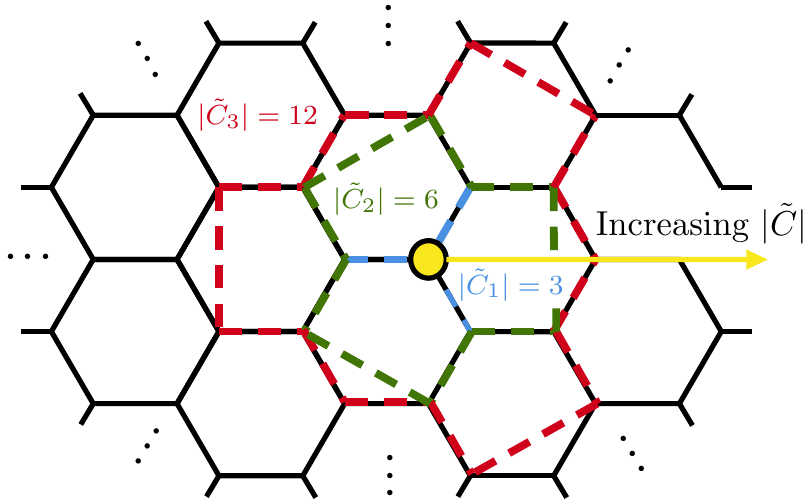}
\caption{Cut-set cardinality with respect to increasing distance from user-node on a honeycomb lattice. We show some example cuts on a honeycomb network of increasing cut-set dimension. The further one moves from a user-node $\bs{a}$, the more edges that must be cut due to $k$-regularity. $\tilde{C}_1$ gives the neighbourhood cut-set $E_{\bs{a}}$, $\tilde{C}_2$ gives the smallest cut-set when limited to network-body edges $E^{\prime}$, and $\tilde{C}_3$ gives a wider cut example.    }
\label{fig:Cuts}
\end{figure}

\begin{theorem}
Consider a $(k,\bs{\Lambda})$-WR quantum network. Select an end-user pair $\bs{i} = \{\bs{a},\bs{b}\}$, and demand they are sufficiently distant such that they do not share an edge or neighbour. Then there exists the threshold single-edge capacity $\mc{C}_{\min}^{\prime}$ in the network bulk, and another for the user-connected edges $\mc{C}_{\min}^{\bs{i}}$ given by
\begin{gather}
\mc{C}_{\min}^{\prime}\defeq \frac{1}{\delta}\>\mc{C}_{\mc{N}_{\bs{i}}}^m, 
~~\mc{C}_{\min}^{\bs{i}}\defeq \left(\frac{1}{k-1} - \frac{1}{\delta}\right)\>\mc{C}_{\mc{N}_{\bs{i}}}^m,\label{eq:NeighThreshs}
\end{gather}
such that if all single-edge capacities in the network satisfy their minimum thresholds, $\mc{C}_{\bs{xy}} \geq \mc{C}_{\min}^{\prime},\forall\>(\bs{x},\bs{y})\in E^{\prime}$ and $\mc{C}_{\bs{xy}} \geq \mc{C}_{\min}^{\bs{i}},\forall\>(\bs{x},\bs{y})\in E_{\bs{a}}\cup E_{\bs{b}}$ then flooding capacity is guaranteed to satisfy
\begin{equation}
\mc{C}^m(\bs{i},\mc{N}) = \mc{C}_{\mc{N}_{\bs{i}}}^{m}.
\end{equation}
\label{theorem:Thresh2}
\end{theorem}

\begin{proof}
By Theorem~\ref{theorem:Thresh1} we know that if all edges satisfy $\mc{C}_{\bs{xy}} \geq \mc{C}_{\min}^{\prime} = \mc{C}_{\mc{N}_{\bs{i}}}^m/\delta$, for all $(\bs{x},\bs{y})\in E$, then the flooding capacity satisfies Eq.~(\ref{eq:PerfBounds}). If we want to avoid the worst-case lower-bound it is possible to enforce an additional, slightly stricter constraint on the end-user connected edges, which we label $\mc{C}_{\min}^{\bs{i}}$. 
If we consider the hybrid cut $C^{\prime\prime}$ scenario as in the previous theorem where $(k-1)$ edges from the user-neighbourhood are collected and have minimum capacity $\mc{C}_{\min}^{\bs{i}}$ along with $(k-1)$ network-bulk edges with capacity $\mc{C}_{\min}^{\prime}$ in order to consolidate the end-user partition. This results in a possible multi-edge capacity
\begin{equation}
\mc{C}^m(C^{\prime\prime}) \geq (k-1)\mc{C}_{\min}^{\prime} + (k-1) \mc{C}_{\min}^{\bs{i}} = \frac{(k-1)}{\delta} \mc{C}_{\mc{N}_{\bs{i}}}^m + (k-1) \mc{C}_{\min}^{\bs{i}}.
\end{equation}
To ensure that $\mc{C}^m(C^{\prime\prime}) \geq \mc{C}_{\mc{N}_{\bs{i}}}^m$ we must then demand that
\begin{align}
&\frac{(k-1)}{\delta} \mc{C}_{\mc{N}_{\bs{i}}}^m + (k-1) \mc{C}_{\min}^{\bs{i}} \geq \mc{C}_{\mc{N}_{\bs{i}}}^m \implies \mc{C}_{\min}^{\bs{i}} \geq \left(\frac{1}{k-1} - \frac{1}{\delta}\right)\mc{C}_{\mc{N}_{\bs{i}}}^m.
\end{align}

Therefore, if we demand that all edges in the user-neighbourhoods satisfy $\mc{C}_{\bs{xy}} \geq \mc{C}_{\min}^{\bs{i}}$, then the worst-case scenario which generates the lower-bound in Theorem~\ref{theorem:Thresh1} disappears and becomes equivalent to the min-neighbourhood capacity. That is, the inequalities become
\begin{equation}
\mc{C}_{\mc{N}_{\bs{i}}}^{m}\leq \mc{C}^m(\bs{i},\mc{N}) \leq \mc{C}_{\mc{N}_{\bs{i}}}^{m} \implies \mc{C}^m(\bs{i},\mc{N}) = \mc{C}_{\mc{N}_{\bs{i}}}^{m} ,
\end{equation}
as required. We find that $\mc{C}_{\min}^{\bs{i}} \geq \mc{C}_{\min}^{\prime}$ if ${k} \leq \frac{\delta}{2} + 1$, which is satisfied in all of our example architectures. 
\end{proof}

\subsection{Neighbour Sharing End-Users}

So far we have considered end-user pairs that are not directly connected and do not share common neighbours. This is  appropriate assumption since we are studying global quantum communications over very long distances; it is not interesting to consider short range users separated by single links. Furthermore, it allows for much clearer intuition surrounding increasing cut-set dimension with respect to cuts on the network-bulk as shown in Fig.~\ref{fig:Cuts}. This assumption does not compromise the generality of our arguments, as we show in the following corollary that Theorems~\ref{theorem:Thresh1} and \ref{theorem:Thresh2} hold even when end-user nodes share a neighbour. 

\begin{corollary} \text{\emph{(Neighbour Sharing):}} Consider a $(k,\bs{\Lambda})$-WR quantum network and an end-user pair $\bs{i} = \{\bs{a},\bs{b}\}$ within the network that do not share an edge, and possess a min-neighbourhood capacity $\mc{C}_{\mc{N}_{\bs{i}}}^{m}$. Even if the end-user nodes share a neighbour, Theorems~\ref{theorem:Thresh1} and \ref{theorem:Thresh2} hold.
\end{corollary}

\begin{proof} 
Consider the $(k,\bs{\Lambda})$-WR network and assume that the end-user pair $\bs{i} = \{\bs{a},\bs{b}\}$ are not directly connected, but share a neighbour. The number of common neighbours that these non-adjacent nodes share is defined by the non-adjacent commonality, $\mu(\bs{a},\bs{b}) > 0$. The previous analyses do not directly apply since cuts restricted to the network-bulk will not be able to partition the two users. This is true because there will exist clear paths along the edges connected to the common neighbours of $\bs{a}$ and $\bs{b}$. Hence, a valid network-cut of these end-users requires one to collect $\mu(\bs{a},\bs{b})$ edges from a user-neighbourhood.\par 

Nonetheless, our results still hold. Let us locate a network-cut that uses the minimum number of user-connected edges possible. This can be considered as a modification to the network-bulk cut which is necessary due to neighbour-sharing.
This cut $C^{\prime}$ still collects at least $\sum_{\lambda\in\bs{\lambda}^*}  (k-\lambda-1)$ edges, but now $\mu(\bs{a},\bs{b})$ of those edges are actually contained in one of the user-neighbourhoods. 
In a worst-case scenario, one may assume that the user-connected edges which are necessarily cut possess the minimum single-edge capacity in the user-neighbourhoods, defined as
\begin{equation}
\mc{C}_{\min}^{\bs{i}} = \min_{\bs{j}\in\{\bs{a},\bs{b}\}}\min_{(\bs{x},\bs{y}) \in E_{\bs{j}}} \mc{C}_{\bs{xy}}.
\end{equation}
Let all edges in the network obey a threshold capacity $\mc{C}_{\min} = \mc{C}_{\mc{N}_{\bs{i}}}^m/\delta$ as motivated by Theorem~\ref{theorem:Thresh1} for non-neighbour sharing end-users. 
Now, the network-cut $C^{\prime}$ which collects the fewest number of edges from the user-neighbourhood will generate a multi-edge capacity,
\begin{align}
\mc{C}^m(C^{\prime}) &\geq {\delta} \mc{C}_{\min} + \mu(\bs{a},\bs{b})( \mc{C}_{\min}^{\bs{i}}- \mc{C}_{\min}  ). \label{eq:NS_Cap}
\end{align}
However, we already know that $\mc{C}_{\min}^{\bs{i}} = \mc{C}_{\min}$ since we stated that all edges in the network obey the same minimum threshold. Therefore,
\begin{align}
&\mc{C}^m(C^{\prime}) \geq {\delta}  \mc{C}_{\min} = \mc{C}_{\mc{N}_{\bs{i}}}^m.
\end{align}
as required. Therefore, neighbour sharing does undermine the previous threshold theorems. Indeed, introducing a stricter condition on the user-connected edges (as we do in Theorem~\ref{theorem:Thresh2}) only makes this result stronger, since $\mc{C}_{\min}^{\bs{i}} \geq \mc{C}_{\min}$ will only increase the multi-edge capacity in Eq.~(\ref{eq:NS_Cap}).
\end{proof}

\subsection{Bosonic Lossy Weakly-Regular Networks\label{sec:BosonicLossy}}

When considering fibre-based networks, point-to-point links are described by bosonic pure-loss (lossy) channels. A lossy channel $\mc{L}$ with transmissivity $\eta \in (0,1)$ is a phase-insensitive Gaussian quantum channel, which transforms input quadratures $\hat{\bf{x}} = (\hat{q},\hat{p})^T$ according to $\hat{\bf{x}} \mapsto \sqrt{\eta} \hat{\bf{x}} + \sqrt{1-\eta}~\hat{\bf{x}}_{\text{env}}$ (where the environment is in a vacuum state) describing the interaction of bosonic mode with a zero-temperature bath \cite{GaussRev}. \par
For lossy quantum networks, the most important property is channel length, or from a network perspective, \textit{inter-nodal separation}. For a given edge $(\bs{x},\bs{y})\in E$ connecting two users in a network, the inter-nodal separation is simply the distance $d_{\bs{xy}}$ between them. All two-way assisted quantum and private capacities of the lossy channel are precisely known via the PLOB bound \cite{PLOB},
\begin{equation}
\mc{C}_{\mc{L}}(d_{\bs{xy}}) = -\log_2\left( 1-10^{-\gamma d_{\bs{xy}} }\right),
\end{equation}
where the inter-nodal separation is related to the transmissivity via $\eta_{\bs{xy}} = 10^{-\gamma d_{\bs{xy}}}$. For current, state of the art fibre-optics the loss rate is $\gamma = 0.02$ per km (which equates to a loss rate of $0.2$ dB/km). Since these separations directly dictate the channel quality between nodes they must be precisely engineered and distributed in order to guarantee strong end-to-end performance.\par

The direct application of Theorem~\ref{theorem:Thresh1} to bosonic lossy quantum networks allows us to translate the notion of threshold capacities into something more physical. Indeed, since the capacity of pure-loss channels is known exactly, it is possible to translate the threshold capacity into a \textit{maximum inter-nodal separation}.\\

\begin{corollary}
Consider a $(k,\bs{\Lambda})$-WR quantum network which is connected by bosonic lossy channels. Select an end-user pair $\bs{i} = \{\bs{a},\bs{b}\}$ within the network that do not share an edge, and possess a min-neighbourhood capacity $\mc{C}_{\mc{N}_{\bs{i}}}^{m}$.
Then, there exists a maximum inter-nodal separation for all edges within the network,
\begin{equation}
d_{\mc{N}}^{\max} = -\frac{1}{\gamma} \log_{10}\left(1 - 2^{-\frac{1}{\delta} \mc{C}_{\mc{N}_{\bs{i}}}^m}\right), \label{eq:Dmax_MNC}
\end{equation}
for which the flooding capacity satisfies
\begin{equation}
\frac{2(k-1)}{\delta} \mc{C}_{\mc{N}_{\bs{i}}}^{m}\leq \mc{C}^m(\bs{i},\mc{N}) \leq \mc{C}_{\mc{N}_{\bs{i}}}^{m}. \label{eq:PerfBounds}
\end{equation}
If  $~\exists\> d_{\bs{xy}} < d_{\mc{N}}^{\max},~(\bs{x},\bs{y})\in E$, this remains an upper-bound on the optimal network performance, $\mc{C}^m(\bs{i},\mc{N}) \leq \mc{C}_{\mc{N}_{\bs{i}}}^{m}$.
\label{corollary:FibThresh1}
\end{corollary}

\begin{proof} Consider a valid pair of end-users $\bs{i} = \{\bs{a},\bs{b}\}$ embedded within a $(k,\bs{\Lambda})$-WR quantum network. Then as before, there exists a threshold capacity $\mc{C}_{\min} = \frac{1}{\delta} \mc{C}_{\mc{N}_{\bs{i}}}^m$ that can be enforced to ensure the flooding capacity between these users is bounded by their min-neighbourhood capacity, $\mc{C}_{\mc{N}_{\bs{i}}}^m$. Supplanting the PLOB bound into the capacity condition in Theorem~\ref{theorem:Thresh1},
\begin{equation}
\mc{C}_{\mc{L}}(d_{\bs{xy}}) \geq \mc{C}_{\min},  \forall (\bs{x},\bs{y}) \in E,
\end{equation}
this readily translates to,
\begin{align}
d_{\bs{xy}} &\leq  -\frac{1}{\gamma} \log_{10}\left(1 - 2^{- \frac{1}{\delta} \mc{C}_{\mc{N}_{\bs{i}}}^m}\right), \forall (\bs{x},\bs{y}) \in E.
\end{align}
Therefore the threshold capacity becomes an upper-bound on the maximum link-length permitted within the network. We can thus define this maximum length,
\begin{equation}
d_{\mc{N}}^{\max} = -\frac{1}{\gamma} \log_{10}\left(1 - 2^{-\frac{1}{\delta} \mc{C}_{\mc{N}_{\bs{i}}}^m}\right),
\end{equation}
which when satisfied ensures that $2(k-1)\mc{C}_{\mc{N}_{\bs{i}}}^m/\delta \leq \mc{C}^m(\bs{i},\mc{N}) \leq \mc{C}_{\mc{N}_{\bs{i}}}^m$.

Now suppose that there exists a channel within the network-bulk that violate this max-bulk separation, i.e.~$\exists\> d_{\bs{xy}} > d_{\mc{N}^{\prime}}^{\max}$ for $(\bs{x},\bs{y}) \in E^\prime$. This violates the threshold capacity condition from Theorem~\ref{theorem:Thresh1} meaning that the minimum-cut in the network is not guaranteed to satisfy the performance bounds. However, if the minimum-cut undergoes a transition due to the introduction of poor quality channels in the network-bulk, it cannot improve the network flooding capacity; it can only deteriorate network performance. Therefore the min-neighbourhood capacity remains an upper-bound on the optimal network performance, $\mc{C}^m(\bs{i},\mc{N}) \leq \mc{C}_{\mc{N}_{\bs{i}}}^m$, as before.
\end{proof}\\

In order to achieve a stricter performance guarantee, we can apply Theorem~\ref{theorem:Thresh2} to bosonic lossy channels and derive slightly stricter constraints on user-connected channels. 

\begin{corollary}
Consider a $(k,\bs{\Lambda})$-WR quantum network which is connected by bosonic lossy channels. Select an end-user pair $\bs{i} = \{\bs{a},\bs{b}\}$ within the network that do not share an edge, and possess a min-neighbourhood capacity $\mc{C}_{\mc{N}_{\bs{i}}}^{m}$.
Then, there exists a maximum link-length in the network-bulk
\begin{equation}
d_{\mc{N}}^{\max} = -\frac{1}{\gamma} \log_{10}\left(1 - 2^{-\frac{1}{\delta} \mc{C}_{\mc{N}_{\bs{i}}}^m}\right),
\end{equation}
and a maximum link-length for all the user-connected edges,
\begin{equation}
d_{\mc{N}_{\bs{i}}}^{\max} = -\frac{1}{\gamma} \log_{10}\left(1 - 2^{-\left(\frac{1}{k-1} - \frac{1}{\delta} \right)\mc{C}_{\mc{N}_{\bs{i}}}^m}\right) \leq d_{\mc{N}}^{\max},
\end{equation}
which when satisfied guarantee that the flooding capacity is equal to the min-neighbourhood capacity,
\begin{equation}
\mc{C}^m(\bs{i},\mc{N}) = \mc{C}_{\mc{N}_{\bs{i}}}^{m}.
\end{equation}
If  $~\exists\>d_{\mc{N}_{\bs{i}}}^{\max} < d_{\bs{xy}} \leq d_{\mc{N}}^{\max},~(\bs{x},\bs{y})\in E_{\bs{a}}\cup E_{\bs{b}}$, we regain Theorem \ref{theorem:Thresh2}. If  $~\exists\>d_{\bs{xy}} > d_{\mc{N}}^{\max},~(\bs{x},\bs{y})\in E $, then the performance guarantee is violated, but this remains an upper-bound on the optimal network performance, $\mc{C}^m(\bs{i},\mc{N}) \leq \mc{C}_{\mc{N}_{\bs{i}}}^{m}$.
\label{corollary:FibThresh2}
\end{corollary}

\begin{proof} 
This is a specification of Theorem \ref{theorem:Thresh2} to bosonic lossy channels where the edges in the neighbourhoods of Alice $\bs{a}$ and Bob $\bs{b}$ may possess their own, stricter constraint in order to completely guarantee optimal performance. The proof follows directly by supplementing the PLOB bound into the threshold capacity expressions. 
\end{proof}

\section{Nodal Densities and Bosonic Lossy Weakly-Regular Networks\label{sec:NDs}}

\subsection{Sparse Constructions}

Nodal density is defined as the number of nodes $n$ per unit area of the network, 
\begin{equation}
\rho_{\mc{N}} \defeq n/A
\end{equation} 
where $A$ is some area in which the network is defined. This is a crucial measure of network resources, especially for quantum networks where there is a very high cost of constructing quantum devices at every node. In many network settings, it is desirable to minimise the nodal density \text{necessary}  to promise strong end-to-end performance. For this reason, it is also useful to define a \textit{minimum nodal density}. For a class of network, $\textsf{{N}} = \{\mc{N}_{j}\}_{j}$, such that all instances $\mc{N} \in \textsf{N}$ are constrained to some implicit structure, the minimum nodal density describes how it can be constructed in the \text{sparsest way possible}. It refers to a limiting scenario in which the network is least dense, and that all other instances of the network topology will possess more nodes per unit area. This is summarised below in a general definition:

\begin{defin} \emph{(Sparse Construction):} Consider a class of network $\textsf{\emph{N}}= \{\mc{N}_{j}\}_{j}$ which imposes a fixed, single-edge distance constraint on its networks $\mc{N} = (P,E) \in \emph{\textsf{N}}$ so that 
\begin{equation}
d_{\bs{xy}} \leq d_{\mc{N}}^{\max} \text{ \emph{for all} } (\bs{x},\bs{y})\in E.
\end{equation}
The sparse construction is an instance of this class which minimises its network nodal density,
\begin{equation}
\rho_{\mc{N}}^{\min} =  \min_{\mc{N} \in \emph{\textsf{N}}} \rho_{\mc{N}} =  \min_{\mc{N} \in \emph{\textsf{N}}} \frac{n}{A},
\end{equation}
where $\rho_{\mc{N}}^{\min}$ is the minimum permitted nodal density of a network $\mc{N}\in \textsf{\emph{N}}$.
\end{defin}

Clearly, for very general classes of distance-constrained networks this minimisation is extremely difficult. However, for analytical classes such as WRNs, this becomes rather easy and reduces to a geometric packing problem.

Finally, to provide simplifications in subsequent arguments, we make the following proposition.
\begin{proposition} 
For a class of single-edge distance constrained networks $\mc{N} \in \textsf{\emph{N}}$, such that $d_{\bs{xy}} \leq d_{\mc{N}}^{\max}$ for all $(\bs{x},\bs{y})\in E$, then the minimum nodal density can always be expressed as
\begin{align}
\rho_{\mc{N}}^{\min} &\propto ({d_{\mc{N}}^{\max}})^{-2} ={\xi}{(d_{\mc{N}}^{\max})^{-2}} \label{eq:MinDensityBound},
\end{align}
such that $\xi$ is a quantity which characterises the network class $\textsf{\emph{N}}$. 
\end{proposition}

It is always possible to express the min-density as proportional to the inverse squared value of the maximum inter-nodal separation in the network. This is obviously true for \textit{any measure of area} since $\rho \propto A^{-1}$ and $A \propto d^{2}$ where $d$ is some distance measure. Yet, for what follows we find that it is useful to closely relate $d_{\mc{N}}^{\max}$ and $\rho_{\mc{N}}^{\min}$ in this way.

\subsection{Sparse Constructions of Weakly-Regular Networks}

In this section we endeavour to lower-bound the min-nodal densities for the classes of WRN.

\subsubsection{Honeycomb Network}

Our model of a honeycomb network ($k=3$ and $\bs{\lambda}^* = \{0\}^{\cup 3}$) is the following: Consider a single, initial hexagon consistent of $n$-nodes connected by $6$ edges. Let us call this the $r=1$ ring of the network. To construct a larger network, we proceed by adding further hexagons concentrically around the initial shape. Each edge of the $r=1$ ring is used as an edge of a hexagon in the $r=2$ ring. We can continue to create a larger and larger network structure by concentrically connecting rings of hexagons to the previous one. As each ring is added, there will be $6r$ hexagons added to the overall structure. See Fig.~\ref{fig:min_nodes_check}(a) as an example.

For any fixed number of rings $r$ we can identify the number of nodes within the network. The number of unique nodes added with the addition of each new ring follows a recursive equation
\begin{align}
\tilde{n}_1 = 6,~\tilde{n}_2 = 12, ~, \ldots , \tilde{n}_r = 24(r-1) - \tilde{n}_{r-1}.
\end{align}
It is simple to solve this set of recursive equations so that $\tilde{n}_r$ takes the form,
\begin{align}
\tilde{n}_r = 6(2r-1).
\end{align}
As a result, given an $r$-ring honeycomb network structure, the total number of nodes will be
\begin{align}
n_{\text{hc}}(r) &=  \sum_{k=1}^r  6(2k-1) = - 6r + 12 \sum_{k=1}^r k = 6 r^2.
\end{align}
The minimum number of nodes required to locate a pair of non-edge sharing end-users within an internal boundary is found at $r=2$. Hence the minimum number of nodes we must consider is simply $n_{\text{hc}}(2) = 24$. 

We may also use this relationship in order to determine the minimum nodal-density $\rho_{\min}$ of a honeycomb network when the maximum permitted fibre-length is $d_{\mc{N}}^{\max}$. Since this is the maximum permitted length and a honeycomb lattice can form a regular tiling, then $\rho_{\min}$ is satisfied when every edge in the network is exactly $d_{\mc{N}}^{\max}$. Hence, given an $r$-ring network, the maximum area it will span is
\begin{align}
A_{\text{hc}}^{\max} (r,d_{\mc{N}}^{\max}) &=  \frac{3\sqrt{3} \left[ 1 + 3r(r+1) \right] {d_{\mc{N}}^{\max}}^2}{2}.
\end{align}
Hence, an $r$-ring minimum nodal density can be computed by
\begin{equation}
\rho_{\min}(r,d_{\mc{N}}^{\max})= \frac{n_{\text{hc}}(r)}{A_{\text{hc}}^{\max}(r,d_{\mc{N}}^{\max})}.
\end{equation}
By taking the asymptotic limit of $r\rightarrow \infty$ we can more accurately capture a lower-bound on the nodal density of a honeycomb network which satisfies this fibre-length constraint (as a larger network will permit a more accurate averaging process). As a result, we may compute
\begin{equation}
\rho_{\mc{N}}^{\min} \geq  \lim_{r\rightarrow \infty} \rho_{\text{min}}(r,d_{\mc{N}}^{\max}) = \frac{4}{3\sqrt{3} \> {d_{\mc{N}}^{\max}}^2}
\end{equation}
as a lower-bound on the nodal-density of a weakly-regular honeycomb network which satisfies a maximum inter-nodal separation. Hence the characteristic quantity of honeycomb networks satisfies $\xi \geq 4/(3\sqrt{3})$.

\subsubsection{Hexagonal Network}

A class of hexagonal network ($k=6$ and $\bs{\lambda}^* = \{2\}^{\cup 6}$) follows the same logic as the honeycomb structure, just with additional nodes located within every hexagon (see Fig.~\ref{fig:min_nodes_check}(b)). As a result, we can immediately write
\begin{equation}
\tilde{n}_r = 6(2r-1) + 6(r-1) = 6(3r-2).
\end{equation}
Then, in an $r$-ring hexagonal structure the total number of nodes is given by,
\begin{equation}
{n}_{\text{hex}}(r) = 7 + 6 \sum_{k=2}^{r}( 3k-2)= 1 + 3r(3r-1).
\end{equation}
hence the minimum number of nodes required for internal WR is $n_{\text{hex}}(2) = 31$. Meanwhile, the maximum area spanned by an $r$-ring hexagonal network is equal to that of the honeycomb network, $A_{\text{hex}}^{\max} = A_{\text{hc}}^{\max}$. Thus, defining
\begin{equation}
\rho_{\text{hex}}(r,d_{\mc{N}}^{\max}) \defeq \frac{n_{\text{hex}}(r)}{A_{\text{hex}}^{\max}(r,d_{\mc{N}}^{\max})},
\end{equation}
we can easily compute a lower-bound on the nodal density as before 
\begin{equation}
\rho_{\mc{N}}^{\min} \geq \lim_{r\rightarrow \infty} \rho_{\text{hex}}(r,d_{\mc{N}}^{\max})= \frac{2}{\sqrt{3}\>{d_{\mc{N}}^{\max}}^2}.
\end{equation}
Hence the characteristic quantity of hexagonal networks satisfies $\xi \geq 2/\sqrt{3}$.

\subsubsection{Manhattan-Inspired Networks}

Consider a class of WRN such that $k=8$ and $\bs{\lambda}^* = \{2,4\}^{\cup 4}$, as depicted in Fig.~\ref{fig:min_nodes_check}(c). To construct this network, we can simply concatenate a cell (consisting of 9 nodes) into an $r \times r$ grid which can be easily evaluated. For a network which is arranged into a $r$-length square grid there will exist $r^{2}$ network cells. In order to maximise the area spanned by each network cell, we assign the longest possible edge in the cell to be of length $d_{\mc{N}}^{\max}$. For the $k=8$ network cell, this means that the diagonal edges in each square must be of length $d_{\mc{N}}^{\max}$. Hence, the area of the total 9 node cell will be $\frac{1}{2}{d_{\mc{N}}^{\max}}^2$. In an $r$-ring network this results in a total area of $\frac{1}{2}(r d_{\mc{N}}^{\max})^2$.
Furthermore, the total number of nodes will be given by,
\begin{equation}
n_{\text{mh:8}}(r) = (r+1)^2,
\end{equation}
which can be obtained by simply counting the number of nodes on each horizontal/vertical row of the grid. We can thus define the function, 
\begin{equation}
\rho_{\text{mh:8}}(r,d_{\mc{N}}^{\max}) \defeq \frac{(r+1)^2}{\frac{1}{2}(r d_{\mc{N}}^{\max})^2}.
\end{equation}
As a result, a lower-bound on the minimum nodal density can be readily computed 
\begin{equation}
\rho_{\mc{N}}^{\min} \geq  \lim_{r\rightarrow \infty} \rho_{\text{mh:8}} (r,d_{\mc{N}}^{\max})= \frac{2}{{d_{\mc{N}}^{\max}}^2}.
\end{equation}
Therefore the characteristic quantity is lower-bounded by $\xi \geq 2$.

A similar Manhattan-like class can be constructed such that $k=16$ and $\bs{\lambda}^* = \{4,8,8,8\}^{\cup 4}$, as depicted in Fig.~\ref{fig:min_nodes_check}(d). Using this network cell to construct a larger network, we must constrain the longest edge in the network cell to be of length $d_{\mc{N}}^{\max}$. This causes us to constrain the diagonal edge from central nodes on the boundary of the cell to connected nodes at the opposite corner. The maximum area spanned by a network cell is then $ \frac{4}{3} d_{\mc{N}}^2$. If we again consider an $r\times r$ cell square grid network, then that the total area is $A_{\max}(r) = \frac{4}{3} r^2 {d_{\mc{N}}^{\max}}^2$. 
Via a counting argument, the total number of nodes in an $r$-radius network will be
\begin{equation}
n(r) = (4r+1)(r+1) + 3r(r+1) + r^2 = (7r+1)(r+1) +r^2.
\end{equation}
As a result, we can define the minimum nodal density function, 
\begin{equation}
\rho_{\text{mh:16}}(r,d_{\mc{N}}^{\max}) \defeq \frac{(7r+1)(r+1)+r^2}{\frac{4}{3}(r d_{\mc{N}}^{\max})^2}.
\end{equation}
Finally, the lower-bound can be given
\begin{equation}
\rho_{\mc{N}}^{\min} \geq \lim_{r\rightarrow \infty} \rho_{\text{mh:16}}(r,d_{\mc{N}}^{\max})= \frac{6}{{d_{\mc{N}}^{\max}}^2}.
\end{equation}
Hence the characteristic quantity can be lower-bounded by $\xi \geq 6$.

\subsection{Nodal Density and End-to-End Performance}

\begin{theorem}
Consider a $(k,\bs{\Lambda})$-WR quantum network $\mc{N} = (P,E)$ which is connected by bosonic lossy channels. Select an end-user pair $\bs{i} = \{\bs{a},\bs{b}\}$ within the network that do not share an edge, and a desired min-neighbourhood capacity $\mc{C}_{\mc{N}_{\bs{i}}}^{m}$. In order to guarantee optimal performance, there exists a minimum nodal density within the network,
\begin{equation}
\rho_{\mc{N}}^{\text{\emph{min}}} = -\xi \gamma^2 \left[ \log_{10}\left(1 - 2^{-\frac{1}{\delta} \mc{C}_{\mc{N}_{\bs{i}}}^m}\right) \right]^{-2}, 
\end{equation}
where $\xi$ is characteristic of the WR architecture being considered.
\label{theorem:Thresh_density}
\end{theorem}

\begin{proof} 
In Corollaries~\ref{corollary:FibThresh1} and \ref{corollary:FibThresh2}, a global fibre-length constraints are placed on the network in order to guarantee a particular flooding capacity via user-node isolation. Using Corollary~\ref{corollary:FibThresh1}, if all edges $(\bs{x},\bs{y})\in E$ satisfy an maximum link-length constraint,
\begin{equation}
d_{\mc{N}}^{\max} = -\frac{1}{\gamma} \log_{10}\left(1 - 2^{-\frac{1}{\delta} \mc{C}_{\mc{N}_{\bs{i}}}^m}\right),
\end{equation}
then the flooding capacity is guaranteed to satisfy $2(k-1)\mc{C}_{\mc{N}_{\bs{i}}}^{m}/\delta \leq \mc{C}^m(\bs{i},\mc{N}) \leq \mc{C}_{\mc{N}_{\bs{i}}}^{m}$. If we apply the additional constraint for user-connected edges such that
\begin{equation}
d_{\mc{N}_{\bs{i}}}^{\max} = -\frac{1}{\gamma} \log_{10}\left(1 - 2^{-\left(\frac{1}{k-1} - \frac{1}{\delta}\right) \mc{C}_{\mc{N}_{\bs{i}}}^m}\right),
\end{equation}
then we can guarantee that $\mc{C}^m(\bs{i},\mc{N}) = \mc{C}_{\mc{N}_{\bs{i}}}^{m}$.

These link-length constraints result in a minimum nodal density for the entire network which is easy to investigate via the appropriate sparse construction. Using Eq.~(\ref{eq:MinDensityBound}) we can directly write
\begin{align}
\rho_{\mc{N}}^{\min} = \xi(d_{\mc{N}}^{\max})^{-2} = \xi\gamma^2 \left[\log_{10}\left(1 - 2^{-\frac{1}{\delta} \mc{C}_{\mc{N}_{\bs{i}}}^m}\right)\right]^{-2},
\end{align}
where the characteristic quantity, $\xi$ is derived from the sparse construction. This offers a lower-bound on the necessary nodal density required to guarantee a particular flooding capacity. 

Summarising, in order for the flooding capacity between $\bs{a}$ and $\bs{b}$ will be equal to the min-neighbourhood capacity $\mc{C}^m(\bs{i},\mc{N}) = \mc{C}_{\mc{N}_{\bs{i}}}^{m}$, the nodal density must (at least) satisfy the lower-bound $\rho_{\mc{N}} \geq \rho_{\mc{N}}^{\min}$. 
\end{proof}\\

The tightness of this lower-bound depends on how $\xi$ is derived. Ideally, one would be able to take into the consideration the stricter constraint $d_{\mc{N}_{\bs{i}}}^{\max}$ required to guarantee $\mc{C}^m(\bs{i},\mc{N}) = \mc{C}_{\mc{N}_{\bs{i}}}^{m}$ with equality. Solving a sparse construction with multiple edge constraints is not straightforward, hence one may need to use a lower-bound for $\xi$, as we have in this work. This nonetheless delivers informative bounds on the nodal density required for optimal performance.

\subsection{Critical Nodal Density}
Following recent works which have investigated critical network resources required for effective end-to-end performance on quantum networks, we define a critical nodal density $\rho_{\text{crit}}$ as the network density required to achieve an end-to-end rate of 1 bit per network use. For bosonic lossy networks constructed in a weakly-regular structure, we can derive this value analytically.\\

\begin{corollary}
Consider a $(k,\bs{\Lambda})$-WR quantum network $\mc{N}=(P,E)$ which is connected by bosonic lossy channels. The critical nodal-density of the network is lower-bounded by
\begin{equation}
\rho_{\mc{N}}^{\text{\emph{crit}}} \geq -\xi \gamma^2 \left[ \log_{10}\left(1 - 2^{-\frac{1}{\delta}}\right) \right]^{-2}, \label{eq:Crit}
\end{equation}
where $\xi$ is characteristic of the WR architecture being considered.
\end{corollary}

In Eq.~(\ref{eq:Crit}), recall that $\gamma$ is the fibre-loss rate which takes a typical value of $\gamma \approx 0.02$, and $\delta$ is defined in Eq.~(\ref{eq:delta_defin}) as before. For WRNs explored in this paper we can readily compute their critical nodal densities:
\begin{align}
\begin{aligned}
\rho_{\text{hc}}^{\text{crit}} &\gtrsim \frac{\sqrt{3}}{5625} \left[ \log_{10}\left(1 - 2^{-1/6}\right) \right]^{-2} \approx 3.33\times 10^{-4} \text{ nodes per km}^{2} ,\\
\rho_{\text{hex}}^{\text{crit}} &\gtrsim \frac{\sqrt{3}}{3750} \left[ \log_{10}\left(1 - 2^{-1/18}\right) \right]^{-2} \approx 2.28\times 10^{-4} \text{ nodes per km}^{2} ,\\
\rho_{\text{mh:8}}^{\text{crit}} &\gtrsim \frac{1}{1250} \left[ \log_{10}\left(1 - 2^{-1/32}\right) \right]^{-2} \approx 2.87\times 10^{-4} \text{ nodes per km}^{2} ,\\
\rho_{\text{mh:16}}^{\text{crit}} &\gtrsim \frac{3}{1250} \left[ \log_{10}\left(1 - 2^{-1/128}\right) \right]^{-2} \approx 4.67\times 10^{-4} \text{ nodes per km}^{2} .
\end{aligned}
\end{align}
The critical density does not necessarily decrease with respect to nodal degree; while the Manhattan-inspired networks have larger regular degrees ($k=8,16$) than the hexagonal network ($k=6$), the critical density of the hexagonal network remains smaller. This is reasonably intuitive, since there is clearly a tradeoff between performance, nodal degree and nodal density. 

Importantly, we notice that the minimum required nodal density is of the order $\sim 10^{-4}$ nodes per $\text{km}^2$ which corroborates the results of Ref.~\cite{QuntaoRandQNets} for which the critical nodal density is studied for more general class of random Waxman networks. Accurate assessments of this kind are crucial to ensure that effective and high-performance quantum networks are constructed in the future.

\end{document}